\RequirePackage{pgf} % Weirdly needed, I bet because of the Springer style
\PassOptionsToPackage{pdfa}{hyperref}
\PassOptionsToPackage{noend}{algpseudocode}
\documentclass[pdflatex,sn-basic,authorversion]{sn-jnl}

\usepackage{amsmath}
\usepackage{amsthm}
\usepackage{amssymb}
\usepackage{bookmark} % mostly to avoid a silly hyperref warning caused by the
\usepackage{booktabs}
\usepackage{subcaption}
\usepackage{xcolor} % for the colors in the hypersetup
\hypersetup{
    colorlinks=true,
    linkcolor={red!50!black},
    filecolor={green!50!black},
    citecolor={green!50!black},
    urlcolor={blue!80!black},
}
\usepackage[capitalise]{cleveref}
\crefname{algorithm}{Alg.}{Algs.}
\Crefname{algorithm}{Algorithm}{Algorithms}
\crefname{appendix}{App.}{App.}
\Crefname{appendix}{Appendix}{Appendices}
\crefname{corollary}{Corol.}{Corolls.}
\Crefname{corollary}{Corollary}{Corollaries}
\crefname{conjecture}{Conjecture}{Conjectures}
\Crefname{conjecture}{Conjecture}{Conjectures}
\crefformat{conjecture}{Conjecture~#2#1#3}
\crefmultiformat{conjecture}{Conjectures~#2#1#3}{\ and~#2#1#3}{, #2#1#3}{\ and~#2#1#3}
\crefname{definition}{Def.}{Defs.}
\Crefname{definition}{Definition}{Definition}
\crefname{figure}{Fig.}{Figs.}
\Crefname{figure}{Figure}{Figures}
\crefname{lemma}{Lemma}{Lemmas}
\Crefname{lemma}{Lemma}{Lemmas}
\crefname{proposition}{Prop.}{Props.}
\Crefname{proposition}{Proposition}{Propositions}
\Crefname{section}{Section}{Sections}
\crefname{section}{Sect.}{Sect.}
\crefname{subsection}{Sect.}{Sect.}
\Crefname{subsection}{Section}{Sections}
\crefname{subsubsection}{Sect.}{Sect.}
\Crefname{subsubsection}{Section}{Sections}
\crefname{table}{Table}{Tables}
\Crefname{table}{Table}{Tables}
\crefname{theorem}{Thm.}{Thms.}
\Crefname{theorem}{Theorem}{Theorems}
\usepackage{footnote}
\makesavenoteenv{tabular}
\makesavenoteenv{table}
\usepackage{mathtools} % for DeclarePairedDelimiter, underbracket
\usepackage{multirow}
\usepackage{MnSymbol}
\PassOptionsToPackage{hyphens}{url} % handle long urls
\usepackage{scalerel}
\usepackage{xfrac} % for nice inline fractions
\usepackage{xspace} % fix space in macros

\newcommand{\rev}[1]{#1}

\newcommand{\spara}[1]{\smallskip\noindent{\bf #1}}

\DeclarePairedDelimiter\abs{\lvert}{\rvert} % absolute value
\makeatletter
\let\oldabs\abs%
\def\abs{\@ifstar{\oldabs}{\oldabs*}}
\makeatother
\newcommand{\algo}{\textsc{Alice}\xspace}
\newcommand{\algorso}{\textsc{Alice-A}\xspace}
\newcommand{\algocurve}{\textsc{Alice-B}\xspace}
\newcommand{\algoseq}{\textsc{Alice-S}\xspace}
\newcommand{\bjdm}[1]{\ensuremath{\mathsf{J}_{#1}}\xspace} % BJDM of graph #1
\newcommand{\card}[1]{\ensuremath{\abs{#1}}\xspace} % cardinality of a set
\newcommand{\copiesnum}[1]{\ensuremath{\mathsf{c}\lparen#1\rparen}\xspace} % no. of matrices corresponding to dataset #1
\newcommand{\caternum}[1]{\ensuremath{\mathsf{z}(#1)}\xspace} % number of caterpillars in a graph % chktex 36
\newcommand{\critval}{\ensuremath{\alpha}\xspace} % critical value
\newcommand{\critvaldatnum}{\ensuremath{T}\xspace} % number of datasets to compute the adjusted critical value
\newcommand{\dataset}{\ensuremath{\datasetsym}\xspace} % dataset
\newcommand{\datasetsym}{\ensuremath{\mathcal{D}}\xspace} % symbol for dataset
\newcommand{\dattomat}[1]{\ensuremath{\mathsf{mat}\lparen#1\rparen}\xspace} % set of matrices corresponding to dataset #1
\newcommand{\degree}[1]{\ensuremath{\mathsf{deg}\lparen#1\rparen}\xspace} % degree of vertex #1 % chktex 35
\newcommand{\fis}[2]{\ensuremath{\mathsf{FI}_{#2}(#1)}\xspace} % frequent itemsets in dataset #2 wrt to threshold #1 % chktex 36
\newcommand{\graphs}{\ensuremath{\mathcal{G}}\xspace} % set of multi-graphs
\newcommand{\items}{\ensuremath{\mathcal{I}}\xspace} % alphabet of items
\newcommand{\matrices}{\ensuremath{\mathcal{M}}\xspace} % set of matrices
\newcommand{\mattodat}[1]{\ensuremath{\mathsf{dat}\lparen#1\rparen}\xspace} % function to map matrices to datasets

\newcommand{\msupp}[2]{\rho_{#1}\lparen#2\rparen} % multi-support of #2 in #1
\newcommand{\neigh}[1]{\ensuremath{\Gamma\left(#1\right)}\xspace} % neighborhood
\newcommand{\neighdistr}[1]{\ensuremath{\xi_{#1}}\xspace} % distribution over the neighbors of #1
\newcommand{\odataset}{\ensuremath{{\smash{\nonsmashedodataset}}}\xspace} % observed dataset
\newcommand{\nonsmashedodataset}{\ensuremath{\mathring{\dataset}}\xspace}
\newcommand{\nullmodel}{\ensuremath{\Pi}\xspace} % null model
\newcommand{\nullprob}{\ensuremath{\pi}\xspace} % null probability distribution
\renewcommand{\nullset}{\ensuremath{\mathcal{Z}}\xspace} % null set of datasets
\newcommand{\pvalsym}{\ensuremath{p}\xspace} % symbol for (estimated) p-value
\newcommand{\pval}[1]{\ensuremath{\pvalsym_{#1}}\xspace} % p-value of a dataset
\newcommand{\epval}[1]{\ensuremath{\tilde{\pvalsym}_{#1}}\xspace} % estimated p-value of a dataset
\newcommand{\epvaldatnum}{\ensuremath{T}\xspace} % number of datasets to estimate the p-value
\newcommand{\statdistr}{\ensuremath{\phi}\xspace} % stationary distribution for MH
\newcommand{\suchthat}{\ensuremath{\mathrel{:}}\xspace} % "such that" symbol for set definitions, with the right spacing
\newcommand{\supp}[2]{\ensuremath{\sigma_{#1}(#2)}\xspace} % support of itemset #1 in dataset #2 % chktex 36
\newcommand{\thresh}{\ensuremath{\theta}\xspace} % minimum support threshold
\newcommand{\ourtitle}{\algo~and~the~Caterpillar:
A~More~Descriptive~Null~Model for~Assessing~Data~Mining~Results}
\newcommand{\slfrac}[2]{\raisebox{5pt}{\ensuremath{\left.\raisebox{5pt}{\ensuremath{\displaystyle#1}}\middle/\raisebox{-5pt}{\ensuremath{\displaystyle#2}}\right.}}}

\theoremstyle{thmstyleone}%
\newtheorem{theorem}{Theorem}
\newtheorem{lemma}[theorem]{Lemma}%
\newtheorem{fact}[theorem]{Fact}
\newtheorem{corollary}[theorem]{Corollary}%

\theoremstyle{thmstylethree}%
\newtheorem{definition}{Definition}%

\newenvironment{squishlist}
{\begin{list}{$\bullet$}
 {\setlength{\itemsep}{0pt}
     \setlength{\parsep}{3pt}
     \setlength{\topsep}{3pt}
     \setlength{\partopsep}{0pt}
     \setlength{\leftmargin}{1.5em}
     \setlength{\labelwidth}{1em}
     \setlength{\labelsep}{0.5em} } }
{\end{list}}

\jyear{2023}%
\begin{document}

\title[\algo and the Caterpillar]{\ourtitle}

\author*[1]{\fnm{Giulia} \sur{Preti}}\email{giulia.preti@centai.eu}
\author[1]{\fnm{Gianmarco} \sur{De Francisci Morales}}\email{gdfm@acm.org}
\author[2]{\fnm{Matteo} \sur{Riondato}}\email{mriondato@amherst.edu}
\affil*[1]{\orgname{CENTAI}, \orgaddress{\street{Corso Inghilterra 3}, \city{Turin}, \postcode{10138}, \state{TO}, \country{Italy}}}
\affil[2]{\orgdiv{Dept.\ of Computer Science} \orgname{Amherst College}, \orgaddress{\street{25 East
Drive} \city{Amherst}, \postcode{01002} \state{MA}, \country{USA}}}

\abstract{We introduce novel null models for assessing the results obtained from
  observed binary transactional \rev{and} sequence datasets, using statistical
  hypothesis testing. Our null models maintain more properties of the observed
  dataset than existing ones. Specifically, they preserve the Bipartite Joint
  Degree Matrix of the bipartite (multi-)graph corresponding % chktex 36
  to the dataset, which ensures that the number of caterpillars, i.e., paths of
  length three, is preserved, in addition to other properties considered by other
  models. We describe \algo, a suite of Markov-Chain Monte-Carlo algorithms for
  sampling datasets from our null models, based on a carefully defined set of
  states and efficient operations to move between them. The results of our
  experimental evaluation show that \algo mixes fast and scales well, and that our
  null model finds different significant results than ones previously considered
  in the literature.}


\maketitle

\noindent {\normalfont\small \emph{``One side will make you grow taller, and the
  other side will make you grow shorter.''} --- The Caterpillar, % chktex 38
  \textit{Alice in Wonderland}}

% !TEX root = ../main-kais.tex
\section{Introduction}\label{sec:intro}

Binary transactional datasets and sequence datasets are the object of study in
several areas, from marketing to network analysis, to finance modeling,
processing of satellite images, and more. In genomics, for example, transactions
represent individuals and the items in a transaction represent their gene
mutations. Many fundamental data mining tasks can be defined on them, such as
frequent itemset/sequence mining, clustering, and anomaly detection.

The goal of knowledge discovery from a dataset is not simply to analyze the
dataset, but to obtain \emph{new understanding} of the stochastic, often noisy,
\emph{process that generated the dataset}. Such novel insights can only be
obtained by subjecting the results of the analysis to a rigorous validation,
which allows to separate those results that give new information about the
process from those that are due to the randomness of the process itself. This
kind of validation is actually necessary in many scientific fields, for example
in microbiology and genomics, when the observed dataset represents individuals
with their gene mutations, or protein interactions
\citep{FerkingstadHS15,RelatorTS18,SeseTST14}.

The \emph{statistical hypothesis testing} framework~\citep{LehmannR22}
is a very rigorous validation process for the results obtained from an observed
dataset. Hypotheses about the results are formulated, and then tested by
comparing the results (or appropriate statistics about them) to their
distribution over the \emph{null model}, i.e., a set of datasets enriched with a
user-specified probability distribution (see \cref{sec:prelims:hyptest}), which
contains all and only the datasets that preserve a user-specified subset of the
properties of the observed dataset (e.g., the size, or some cumulative
statistics). The testing of hypotheses requires, in \emph{resampling-based methods}
\citep{WestfallY93}, to be able to efficiently draw multiple datasets
from the null model. These samples are then used to obtain an approximation of
the distribution of results from the null model, to which the actually observed
results are compared. When the probability of obtaining results as or more
extreme than those observed is low, the observed results are deemed
\emph{statistically significant}, i.e., they are deemed to give previously
unknown information about the data-generating process.

Informally, the properties preserved by the null model, and the sampling
distribution, capture the existing or assumed knowledge about the process that
generated the observed dataset. Testing the hypotheses can be understood as
trying to ascertain whether the observed results can be explained by the
existing knowledge. The choice of the null model must be made by the user, based
on their domain knowledge, and should be deliberate. Null models that capture
more properties of the observed dataset are usually more descriptive and
therefore to be preferred. The challenge in using such models is the need for
efficient computational procedures to draw datasets from the null model
according to the user-specified distribution, as many such sampled datasets are
necessary to test complex or multiple hypotheses.

\subsection*{Contributions}\label{sec:intro:contr}

We study the problem of assessing results obtained from an observed
binary\footnote{In the rest of the work, we drop the attribute ``binary'': all
datasets we refer to are binary.} transactional or sequence dataset by
performing statistical hypothesis tests via resampling methods from a
descriptive null model. Specifically, our contributions are the following.

\begin{itemize}
  \item We introduce novel null models (\cref{sec:model,sec:seq:null}) that preserve
    additional properties of the observed dataset than those preserved by
    existing null models \citep{GionisMMT07,TononV19}. Specifically, all datasets in our
    null models have the same \emph{Bipartite Joint Degree Matrix (BJDM)} of the
    bipartite (multi-)graph corresponding to the observed dataset % chktex 36
    (\cref{sec:model:matrices,sec:model:bjdm}). Maintaining the BJDM captures
    additional ``structure'' of the observed dataset: e.g., on transactional
    datasets, in addition to dataset size, transaction lengths, and item or
    itemset supports, the number of \emph{caterpillars} in the observed
    dataset is also preserved (\cref{lem:caternum}). We also explain why more natural properties,
    such as the supports of itemsets of length two on transactional datasets,
    are not as informative as one may think.
  \item We present \algo,\footnote{Like the eponymous character of \textit{Alice
    in Wonderland}, our algorithms explore a large strange world, and interact
  with caterpillars.} a suite of Markov-Chain-Monte-Carlo algorithms for
  sampling datasets from our null models according to a user-specified
  distribution. \algorso (\cref{sec:sampling:swaps}) is based on
  \emph{Restricted Swap Operations (RSOs)} on biadjacency matrices, which
  preserve the BJDM\@. Our contributions include a sampling algorithm to draw
  such RSOs much more efficiently than with the natural rejection sampling
  approach. A second algorithm, \algocurve, (\cref{sec:sampling:curve}) adapts
  the \textsc{CurveBall}
  approach\rev{~\citep{verhelst2008efficient,carstens2015proof}} to RSOs, to
  essentially perform multiple RSOs at every step, thus leading to faster
  mixing. Finally, \algoseq samples from the null model for sequence datasets,
  using Metropolis-Hastings and a variant of RSOs, to take into account the fact
  that the bipartite graph corresponding to a sequence dataset is a
  \emph{multi-graph}.
\item The results of our experimental evaluation show that \algo mixes fast,
  it is scalable as the dataset grows, and that our new null model differs from
  previous ones, as it marks different results as significant.
\end{itemize}

The present article extends the conference version~\citep{PretiDFMR22} in
multiple ways, including:
\begin{itemize}
  \item The extension to sequence datasets and the development of \algoseq
    (\cref{sec:seq}) is
    entirely new. In addition to introducing a novel null model and algorithm, to the best of our knowledge, 
    our work is the first to look at sequence datasets as bipartite
    multi-graphs, which is a generic
    representation that can be used in other works.
  \item We give an explicit counterexample (\cref{fig:counter}) showing that
    preserving the number of caterpillars and other fundamental properties is
    not sufficient to preserve the BJDM, while the opposite is true
    (\cref{sec:model:bjdm}).
  \item We include a discussion of Gram mates~\citep{Kirkland18,KimK22}, to
    explain why a model preserving the supports of itemsets of length two may
    not be very interesting.
  \item We add examples and figures to help the understanding of important
    concepts.
\end{itemize}

\paragraph{Outline} After discussing related work in \cref{sec:related}, we focus
the presentation on binary transactional datasets, with preliminaries
(\cref{sec:intro}) also covering statistical hypothesis testing. Then we
describe the null model for transactional datasets (\cref{sec:model}), and then
the two algorithms to sample datasets from this null model
(\cref{sec:sampling}). Covering first only transactional datasets allows us to
discuss sequence datasets, the null model, and the specific algorithm
for this case in \cref{sec:seq}. Our experimental evaluation and its results are
presented in \cref{sec:exper}.

% !TEX root = ../main-kais.tex
\section{Related Work}\label{sec:related}

The need for statistically validating results from transactional datasets was
understood immediately after the first efficient algorithm for obtaining these
results was introduced \citep{BrinMS97,MegiddoS98}. A long line of works also
studies how to filter out uninteresting patterns, or directly mine
\emph{interesting} ones \citep{VreekenT14}. This direction is orthogonal to the
study of the \emph{statistical validity} of the results, which is our focus.

Many works concentrate on the case of \emph{labeled} transactional datasets \citep{TeradaKS15,TeradaOHTS13,TeradaTS13,PellegrinaRV19a,Hamalainen16,PellegrinaV20,PapaxanthosLLBB16,MinatoUTTS14,LlinaresLopezSPB15,KomiyamaIANM17,WuHGLZY16,DuivesteijnK11}, where each transaction comes with a binary label.
Most of these works use resampling-based approaches, as we do, but the very different nature of the studied tasks and data, as we study the \emph{unlabeled} case, make them inapplicable to our problems.
We refer to the tutorial by \citet{PellegrinaRV19b} for a detailed survey of the work done in \emph{unlabeled datasets}, including resampling methods.
The different nature of the data makes these approaches inapplicable to our case.

Most work has been on mining \emph{significant frequent itemsets}, tiles, or
association rules~\citep{Hamalainen10,Webb07,LijffijtPP14}. The survey by
\citet{HamalainenW19} presents many of these works in depth. The most relevant
to ours are those by \citet{GionisMMT07} and \citet{Hanhijarvi11}, who present
resampling methods for drawing transactional datasets from a null model which
preserves the number of transactions, the transaction lengths, and the item
supports as in an observed dataset. These approaches, like ours, can be used for
testing any result from transactional datasets, not just for significant pattern
mining. We present a null model that is more descriptive than the ones studied
in these works, because it preserves additional properties of the observed
dataset. \Citet{DeBie10} proposes a method to \emph{uniformly} sample datasets
from a null model that preserves, \emph{in expectation}, the same constraints.
While it can partially be extended to preserve the constraints exactly, it
cannot be used to sample according to any user-specified distribution, which
we believe to be a fundamental ingredient of the null model, as it includes
already available knowledge of the data generating process \emph{in addition to}
the constraints.

Assessing results obtained from sequence datasets has also generated
interest~\citep{PinxterenC21,TononV19,JenkinsWGR22}. We refer the interested
reader for an in-depth discussion of related work in this area to~\citep[Sect.\
2]{JenkinsWGR22}. To the best of our knowledge, we are the first to look at
sequence datasets as bipartite \emph{multi-}graphs, and to propose a null model that explicitly preserves properties of such multi-graphs. Our null model for
sequence datasets preserves additional properties than the one introduced by
\citet{TononV19}, similarly to how our null model for transactional datasets
preserves additional properties than the one by \citet{GionisMMT07}, as indeed
the \citeauthor{TononV19}'s model is essentially an adaptation of the
\citeauthor{GionisMMT07}'s model to sequence datasets. \Citet{TononV19} and
\citet{JenkinsWGR22} present other null models for sequence datasets. Extending
these models to preserve the additional properties we consider is an interesting
direction for future work.

Beyond binary transactional and sequence datasets, resampling methods for
assessing data mining results have been proposed for graphs
\citep{HanhijarviGP09,SugiyamaLLKB15,silva2017network,GunnemannDJE12},
real-valued and mixed-valued matrices \citep{Ojala10}, and database tables
\citep{OjalaGGM10}. None of these works proposes a null model similar to the one
we introduce, nor presents similar sampling algorithms. Our approach can be a
starting point to develop more descriptive null models for these richer types of
data.

\algo, our algorithm for sampling from a null model of datasets, can also be
seen as sampling from the set of
% 20220906 MR: While to me the concept of "non-right-isomorphic" makes sense, I
% can't find it mentioned anywhere else in the literature, so let's just not
% say anything at this point, as there is no problem with that.
%non-right-isomorphic
bipartite graphs with a prescribed BJDM, according to a desired sampling
distribution. In this sense, our contributions belong to a long line of works
that studies how to generate (bipartite) graphs with prescribed properties and
according to a desired probability distribution~\citep{cimini2019statistical,bonifati2020graph,greenhill2022generating,akoglu2009rtg,aksoy2017measuring,saracco2015randomizing,karrer2011stochastic,van2021random,fischer2015sampling,ritchie2017generation,silva2017network,orsini2015quantifying,tillman20192k+}.
The surveys by \citet{cimini2019statistical}, \citet{bonifati2020graph},
and \citet{greenhill2022generating} give complete coverage of this field. These
approaches have been studied in the context of complex networks, while we use
\emph{bipartite} graphs to represent transactional datasets, and our main goal
is to statistically assess results obtained from such datasets, not to study the
properties of the graphs.

No previous work on sampling bipartite graphs deals with the question we study.
\Citet{saracco2015randomizing} presents a configuration model to sample bipartite
networks that, \emph{in expectation}, have the same degree sequences as a
prescribed one. \algo \emph{exactly} maintains the BJDM, which preserves the
exact degree sequences, and also other additional properties (see
\cref{sec:model}); thus our null model preserves more characteristics of the
observed dataset. \Citet{aksoy2017measuring} proposes a method to generate
bipartite networks that preserve also the clustering coefficient, which is not
related to the BJDM\@. \Citet{amanatidis2015graphic} gives necessary and
sufficient conditions for a matrix to be the BJDM of a bipartite graph. We
always start from such a matrix, so we do not have to address its
realizability. The concept of \emph{Restricted Swap Operation (RSO)} was
introduced by \citet{czabarka2015realizations}, but not for the purpose used in \algo. 
\Citet{boroojeni2017generating} presents randomized
algorithms to generate a bipartite graph from a BJDM, but there is no proof that
their approaches can generate all possible graphs with that BJDM nor there is an
analysis on the probability that such a graph is generated. Both aspects are
important in order to use the samples for statistical hypothesis testing (see
\cref{sec:prelims:hyptest}), and \algo achieves these goals.

We are interested in sampling graphs (but really, datasets) from a set of graphs that
preserve the same properties as some observed graph (i.e., dataset). This task
is different from the problem of generating a graph from a random family, such
as Erd\H{o}s-R\'{e}nyi graphs, stochastic block models, Kronecker graphs,
preferential attachment graphs, and others, or fitting the parameters of such a
family on the basis of one or more observed graphs.

% 20220405 MR: commenting out as these works are not entirely relevant, but may
% be useful for the future, either for this project or others.
%
%\cite{wang2020fast} introduces the \emph{Rectangle Loop} algorithm, which also
%improves the performance of the swap algorithms, by using a conditional sampling
%scheme: rather than selecting 2 rows and 2 columns, it \textbf{(i)} selects one
%row $R_1$ and one column $C_1$ with entry 1, \textbf{(ii)} selects uniformly at
%random a column $C_2$ among those with entry 0 in $R_1$, \textbf{(iii)} selects
%uniformly at random a row $R_2$ among those with entry 1 in $C_2$, and
%\textbf{(iv)} performs the swap if the entry of $(R_2, C_1)$ is 0.
%
%\cite{erdHos2022mixing} presents some cases in which the Markov chain is rapidly
%mixing.
%
%\cite{amanatidis2022rapid} proves that the switch Markov chain for generating
%random networks with a given JDM is rapidly mixing.
%
%\cite{arman2019fast} states that MCMC algorithms are almost uniform generators,
%and proposes an exactly uniform generator. The algorithm first generates a
%random multigraph, and then switches the double edges to get a simple graph with
%pre-defined degree sequence. The paper shows how the algorithm can be adapted to
%work with bipartite graphs (i.e., two degree sequences in input).

% !TEX root = ../main-kais.tex
\section{Preliminaries}\label{sec:prelims}

We now define the key concepts and notation used in this work.
\Cref{tab:reference} summarizes the most important notation. Preliminaries for
sequence datasets are deferred to \cref{sec:seq:prelims}.

\subsection{Transactional Datasets}\label{sec:prelims:trans}

Let $\items \doteq \{ a_1, \dotsc, a_{\card{\items}} \}$ be a finite alphabet of
\emph{items}. W.l.o.g., we can assume $\items = \{1, \dotsc, \card{\items}\}$.
Any $A \subseteq \items$ is an \emph{itemset}. A \emph{transactional
  dataset}\footnote{From here to the end of \cref{sec:sampling}, we only discuss
    \emph{transactional} datasets, so we drop the attribute and just refer to
  them as ``datasets''.} $\dataset$ is a finite bag of itemsets, which are known
also as \emph{transactions} when considered as the elements of a dataset. The
\emph{size} $\card{\dataset}$ of the dataset is the number of transactions it
contains. The \emph{length} $\card{t}$ of a transaction $t \in \dataset$ is the
number of items in it. \Cref{fig:transformation} (lower) shows a dataset of
shopping baskets with three baskets (transactions) of length $6$, $5$, and $4$,
respectively.

\begin{figure}[htb]
  \centering
  \includegraphics[width=.8\columnwidth]{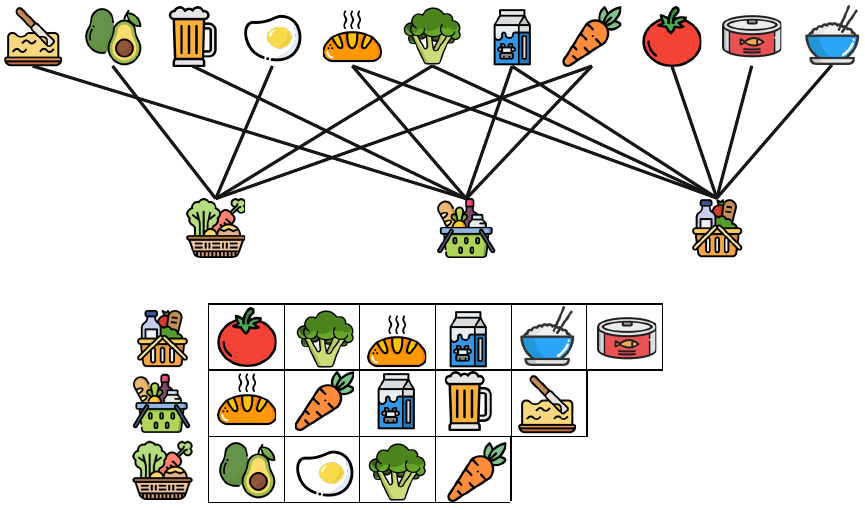}
  \caption{A dataset of shopping baskets (lower) and the respective bipartite graph (upper).}%
  \label{fig:transformation}
\end{figure}

For any itemset $A \subseteq \items$, the \emph{support $\supp{\dataset}{A}$ of $A$
in $\dataset$} is the number of transactions of $\dataset$ which contain $A$:
\[
  \supp{\dataset}{A} \doteq \card{\{t \in
  \dataset \suchthat A \subseteq t\}} \enspace.
\]
The support is a natural (albeit not without drawbacks) measure of
interestingness. A foundational knowledge discovery task requires to find, given
a \emph{minimum support threshold $\thresh \in [0, \card{\dataset}]$}, the
collection $\fis{\dataset}{\thresh}$ of \emph{Frequent Itemsets (FIs) in
$\dataset$ w.r.t.\ $\thresh$}: $\fis{\dataset}{\thresh} \doteq \{ A \subseteq
  \items \suchthat \supp{\dataset}{A} \ge \thresh\}$~\citep{AgrawalS94}.
Given $\thresh = 2$, for \dataset in \cref{fig:transformation} (lower), $\fis{\dataset}{\thresh}$ contains the itemsets $\{\text{ carrot }\}$, $\{\text{ broccoli }\}$, $\{\text{ bread }\}$, $\{\text{ milk }\}$, and $\{\text{ bread }, \text{ milk }\}$.   

\subsection{Null Models and Hypothesis Testing}\label{sec:prelims:hyptest}

The statistical hypothesis testing framework~\citep{LehmannR22} allows
to rigorously understand whether the results obtained from an \emph{observed
dataset $\odataset$} (e.g., the collection of frequent itemsets, or its size,
among many others) are actually interesting or are just due to randomness in the
(unknown, at least partially) data generation process. Informally, the
observed results are compared to the distribution of results that would be
obtained from a \emph{null model} (see below); if results as or more extreme
than the observed ones are sufficiently unlikely, the observed results are
deemed \emph{statistically significant}.

A \emph{null model} $\nullmodel = (\nullset, \nullprob)$ is a pair where
$\nullset$ is a set of datasets, and $\nullprob$ is a (user-specified)
probability distribution over $\nullset$. The datasets in $\nullset$ are all and
only those that share some descriptive characteristics with an \emph{observed
dataset} $\odataset$, which also belongs to $\nullset$.\footnote{Thus,
  $\nullmodel$ depends on $\odataset$, but we hide it in the
notation to keep it light.} Null models in previous
works~\citep{GionisMMT07,DeBie10} preserve the following two \emph{fundamental
properties}:
\begin{squishlist}
  \item the distribution of the transaction lengths, i.e., for any possible
    transaction length $\ell \in [1, \card{\items}]$, $\dataset \in \nullset$
    contains the same number of transactions of length $\ell$ as
    $\odataset$;\footnote{This property implies that the size of the dataset is
      preserved as well, i.e., $\card{\dataset} = \card{\odataset}$ for
    any $\dataset \in \nullset$.} and
  \item the support of the items, i.e., for any $i \in \items$ and $\dataset \in
  \nullset$, $\supp{\dataset}{i} = \supp{\odataset}{i}$.
\end{squishlist}

The intuition behind wanting to preserve some properties of $\odataset$ is that
these properties, together with $\nullprob$, capture what is known or assumed
about the process that generated the data, and the goal is to understand whether
the results obtained from $\odataset$ are, informally, ``typical'' for datasets
with these characteristics. Formally, given $\odataset$ and a null model
$\nullmodel = (\nullset, \nullprob)$, one formulates a \emph{null hypothesis}
$H_0$ involving $\nullmodel$ and a result $R_\odataset$ obtained from
$\odataset$. For example, let $R_\odataset = \card{\fis{\odataset}{\thresh}}$,
and
\begin{equation}\label{eq:nullhyp}
  H_0 \doteq \text{``} \mathop{\mathbb{E}}_{\dataset \sim
  \nullprob}[\card{\fis{\dataset}{\thresh}}] = R_\odataset
  \text{''}.\footnote{This hypothesis is just one simple example of many possible
    different hypotheses that could be tested.}
\end{equation}
The hypothesis is then tested by computing the \emph{$p$-value
$\pval{\odataset,H_0}$ of $H_0$}, defined as the probability that, in a dataset
$\dataset'$ sampled from $\nullset$ according to $\nullprob$, the results
$R_{\dataset'}$ (e.g., $\card{\fis{\dataset'}{\thresh}}$) are \emph{more
extreme} (e.g., larger) than $R_\odataset$, i.e.,
\begin{equation}\label{eq:pval}
  \pval{\odataset, H_0} \doteq \Pr_{\dataset' \sim \nullprob}(
  R_{\dataset'}\ \text{more extreme than}\ R_\odataset) \enspace.
\end{equation}
The notion of ``more extreme'' depends on the nature of
$R_\odataset$. When $\pval{\odataset, H_0}$ is \emph{not larger}
than a user-specified \emph{critical value} $\critval$, then the observed
results $R_\odataset$ are deemed to be \emph{statistically significant}, i.e.,
unlikely to be due to random chance (in other words, the null hypothesis
$H_0$ is rejected as not sufficiently supported by
the available data).
%The probability of wrongly marking $R_\odataset$ as
%significant when it is not, i.e., of making a \emph{false discovery}, is at most
%$\critval$.

Computing the $p$-value $\pval{\odataset,H_0}$ from~\eqref{eq:pval} exactly is
often essentially impossible. \rev{E.g., for statistically-sound knowledge
discovery tasks on
sequence datasets, the exact distribution of test statistics is known only in
very restricted cases~\citep{PinxterenC21}, while all other approaches use
resampling~\citep{TononV19,JenkinsWGR22}.} Thus, an empirical estimate
$\epval{\odataset,H_0}$ is obtained as follows and used in place of
$\pval{\odataset}$ when testing the hypothesis~\citep{WestfallY93}. Let
$\dataset_1, \dotsc, \dataset_{\epvaldatnum}$ be $\epvaldatnum$ datasets
\emph{independently sampled} from $\nullset$ according to $\nullprob$, then
\begin{equation}\label{eq:epval}
  \epval{\odataset,H_0} \doteq \frac{1 + \card{\{ \dataset_i \suchthat
  R_{\dataset_i} \ \text{is more extreme than} \ R_\odataset\}}}{1 +
  \epvaldatnum} \enspace.
\end{equation}
\rev{Such \emph{resampling methods}, of which the well-known bootstrap is also
  an instance, are often to be preferred to the explicit derivation of the statistics
  for multiple reasons:
\begin{itemize}
  \item they are, in some sense, independent from the test being conducted, as
    the test statistic distribution (or better, the $p$-value) is estimated from
    the sampled datasets, as in~\eqref{eq:epval};
  \item they leverage data-dependent distributional characteristics, which tend
    to result in higher statistical power; and
  \item they scale to high-dimensional settings.
\end{itemize}
}

In many knowledge discovery tasks, \rev{and in many applications such as during
clinical trials for drug approvals~\citep{he2021resampling}, or in genomics
studies~\citep{goeman2014multiple}, one is} interested in
testing \emph{multiple hypotheses}. For example, \emph{significant itemset
mining} (see \cref{sec:related}) requires testing one hypothesis
\[
  H_0^{A} \doteq \text{``} \mathop{\mathbb{E}}_{\dataset \sim
  \nullprob}[\supp{\dataset}{A}] = \supp{\odataset}{A} \text{''}
\]
for each itemset $A$.\footnote{This hypothesis is one of many kinds of
hypotheses that can be tested by using the support as the test statistic.}
When testing multiple hypotheses, i.e., all hypotheses in a class $\mathcal{H}$,
one is interested in ensuring that the \emph{Family-Wise Error Rate}, i.e., the
probability of making \emph{any} false discovery, is at most a user-specified
acceptable threshold $\delta$.  Classic methods for controlling the FWER, such
as the Bonferroni correction~\citep{Bonferroni37}, lack the \emph{statistical
power} to be useful in knowledge discovery settings, i.e., the probability that
a \emph{true} significant discovery is marked as such is very low, due to the
large number $\card{\mathcal{H}}$ of hypotheses. \emph{Resampling-based
methods}~\citep{WestfallY93} perform better for these tasks because they
empirically estimate the distribution of the minimum $p$-value of the hypotheses
in $\mathcal{H}$ by \emph{sampling datasets from $\nullset$}, and use this
information to compute an \emph{adjusted critical value} $\hat{\critval}$.

For example, the
Westfall-Young approach works as follows. Let $\dataset'_1, \dotsc,
\dataset'_{\critvaldatnum}$ be $\critvaldatnum$ datasets \emph{sampled
independently} from $\nullset$ according to $\nullprob$, and let
\begin{equation}\label{eq:minpval}
  \check{\pvalsym}_i \doteq \min_{h \in \mathcal{H}} \pval{\dataset'_i,h}
\end{equation}
be the minimum $p$-value, on $\dataset'$, of any hypothesis $h \in \mathcal{H}$.
The \emph{adjusted critical value} $\hat{\critval}$ to which the $p$-values of
the hypotheses are compared is
\[
  \hat{\critval} \doteq \max \left\{ \critval \suchthat \frac{\card{\{ \dataset'_i
  \suchthat \check{\pvalsym}_i \le \critval \}}}{\critvaldatnum} \le \delta
  \right\} \enspace.
\]
That is, $\hat{\critval}$ is the largest $\alpha \in [0,1]$ such that the
fraction of the $\critvaldatnum$ datasets $\dataset'_i$ whose minimum $p$-value
$\check{\pvalsym}$ is at most $\alpha$ is not greater than $\delta$.
Estimates computed as in~\eqref{eq:epval} are used in place of the exact
$p$-values in the r.h.s.\ of~\eqref{eq:minpval}. Comparing the (estimated)
$p$-value of each hypothesis in $\mathcal{H}$ to $\hat{\critval}$ guarantees
that the FWER is at most $\delta$.
Thus,  efficiently drawing random datasets from $\nullset$ according to
$\nullprob$ plays a key role in statistical hypothesis testing.
Our goal in this work is to develop efficient methods to
sample a dataset from $\nullset$ according to $\nullprob$ where $\nullset$ is
the set of datasets that, in addition to preserving the aforementioned three
properties from $\odataset$, also preserve an additional important
characteristic property that we describe in \cref{sec:model:bjdm}.

\subsection{Markov Chain Monte Carlo Methods}\label{sec:prelims:mcmc}

\algo follows the \emph{Markov chain Monte Carlo (MCMC) method}, and uses the
\emph{Metropolis-Hastings (MH) algorithm}~\citep[Ch.\ 7 and 10]{MitzenmacherU05}.
Next is an introduction tailored to our work.

Let $G = (V,E)$ be a directed, weighted, strongly connected, aperiodic
graph, potentially with self-loops. The vertices $V$ are known as \emph{states}
in this context. W.l.o.g., we can assume $V = \{1, 2, \dotsc, \card{V}\}$. For
any state $v$, let $\neigh{v}$ be the set of (out-)neighbors of $v$, % chktex 36
i.e., the set of states $u$ such that $(v,u) \in E$ (it holds $v \in \neigh{v}$
if there is a self-loop). For any neighbor $u \in \neigh{v}$, the weight
$\mathsf{w}(v,u)$ of the edge $(v,u)$ is strictly positive, and it holds
$\sum_{u \in \neigh{v}} \mathsf{w}(v,u) = 1$. In other words, there is a
probability distribution $\neighdistr{v}$ over $\neigh{v}$ such that
$\neighdistr{v}(u) = \mathsf{w}(v,u)$. Let $W$ be the $\card{V} \times \card{V}$
matrix such that $W[v,u] = \mathsf{w}(v,u)$ if $(v,u) \in E$, and $0$
otherwise.\footnote{The strong-connectivity and aperiodicity of $G$, together
  with having $W[u,v] \ge 0$ iff $(u,v) \in E$, ensure that the Markov chain on
$V$ whose matrix of transition probabilities is $W$ has a unique stationary
distribution~\citep[Thm.\ 7.7]{MitzenmacherU05}.}

Let $G = (V,E)$ be a directed, weighted, strongly connected, aperiodic
graph, potentially with self-loops.
The \emph{Metropolis-Hastings (MH) algorithm} gives a way to sample an
element of $V$ according to a user-specified probability distribution
$\statdistr$. Let $v \in V$ be any state, chosen arbitrarily. We first draw a
neighbor $u \in \neigh{v}$ of $v$ according to the distribution
$\neighdistr{v}$. Then we ``move'' from $v$ to $u$ with probability
\begin{equation}\label{eq:mh}
  \min \left\{ 1, \frac{\statdistr(u) \neighdistr{u}(v)}{\statdistr(v)
  \neighdistr{v}(u)} \right\},
\end{equation}
otherwise, we stay in $v$. After a sufficiently large number of
steps $t$, the state $v_t$ is (either approximately or exactly) distributed
according to $\statdistr$ and can be taken as a sample.

In summary, to be able to use MH, one must define the graph $G = (V,E)$, the
neighbor-sampling probability $\neighdistr{v}$ for every $v \in V$, a procedure
to sample a neighbor of $v$ according to $\neighdistr{v}$, and the desired
sampling distribution $\statdistr$ over $V$.

\begin{table}[htbp]
\caption{Table of symbols.}\label{tab:reference}
\small
\begin{center}
\begin{tabular}{cr p{0.8\textwidth}}
\toprule
& \textbf{Symbol} & \textbf{Description} \\
  \midrule
  \parbox[t]{3mm}{\multirow{6}{*}{\rotatebox[origin=c]{90}{Dataset}}}
  & \items          & Set of items \\
  & S               & Ordered list of itemsets \\
  & \dataset        & Dataset (bag of itemsets in the transactional case) \\
  &                 & \textcolor{white}{Dataset} (bag of sequences in the sequence case)\\
  & $M_{\dataset}$  & Binary matrix associated to the transactional dataset \dataset \\
  & $\dattomat{\dataset}$  & Set of binary matrices associated to the transactional dataset \dataset \\
  & $\mattodat{M}$  & Transactional dataset whose binary matrix is $M$ \\
  & \odataset       & Observed dataset \\
  \midrule
  \parbox[t]{3mm}{\multirow{8}{*}{\rotatebox[origin=c]{90}{Bipartite (multi-)Graph}}}
  & $G$             & Bipartite (multi-)graph \\
  & $L \cup R$      & Set of left ($L$) and right ($R$) vertices of $G$ \\
  & $E$             & Set of (multi-)edges of $G$ \\
  & $\graphs$       & Set of bipartite multi-graphs \\
  & $\neigh{v}$     & Set of nodes connected to $v$ in $G$ \\
  & $\bjdm{G}$      & Bipartite Joint Degree Matrix (BJDM) of $G$ \\
  & $\caternum{G}$  & Number of simple paths of length 3 (caterpillars) in $G$\\
  & $\matrices$     & Set of binary matrices of graphs with the same BJDM \\[9pt]
  \midrule
  \parbox[t]{3mm}{\multirow{3}{*}{\rotatebox[origin=c]{90}{Null Model}}}
  & $\nullmodel$    & Null model \\
  & $\nullset$      & Set of datasets sharing some properties of \odataset \\
  & $\nullprob$     & Probability distribution over $\nullset$ \\
  & $\pval{\odataset,H_0}$ & p-value of a null hypothesis $H_0$ involving $\nullmodel$ and $\odataset$\\[3pt]
  % & $\neighdistr{v}$ & Neighbor-sampling probability of $v$ \\
  % & $\supp{\dataset}{A}$    & Bag of itemsets in \dataset containing $A$ \\
  % & $\fis{\dataset}{\thresh}$   & Set of itemsets contained in at least $\thresh$ itemsets in \dataset \\
\bottomrule
\end{tabular}
\end{center}
\end{table}%

% !TEX root = ../main-kais.tex
\section{A More Descriptive Null Model}\label{sec:model}

As discussed in \cref{sec:prelims:hyptest}, a good null model should preserve
important characteristics of the observed dataset $\odataset$, and we mentioned
the two fundamental properties that \rev{were the focus of previous
work}~\citep{GionisMMT07,DeBie10}. We now introduce a null model that preserves
an additional property, and then show efficient methods to sample datasets from
it.

\subsection{Datasets, Matrices, and Bipartite Graphs}\label{sec:model:matrices}

Before defining the additional characteristic quantity of $\odataset$ that we want to
preserve, we must describe
``alternative'' representations of a dataset $\dataset$. The most natural one is a
\emph{binary matrix} $M_\dataset$ with $\card{\dataset}$ rows and
$\card{\items}$ columns, where the $(i,j)$ entry is 1 iff transaction $i \in
\dataset$ contains item $j \in \items$, and where the order of the
transactions (i.e., of the rows) is arbitrary~\citep[Sect.\
4.1]{GionisMMT07}. Since the order is arbitrary, there are \emph{multiple
matrices} that correspond to the same dataset, differing by the ordering of the
rows. This fact is of key importance for the correctness of methods that
sample datasets (and not matrices) from a null model\rev{, i.e., that are 
\emph{row-order agnostic}~\citep{AbuissaLR23}}.
%: \algo uses a set of matrices, which
%contains, for every
%dataset $\dataset \in \nullset$, a (potentially improper, but non-empty)
%\emph{subset} of the matrices corresponding to $\dataset$ (see
%\cref{sec:sampling}).

\emph{Any} matrix $M_\dataset$ corresponding to $\dataset$
can be seen as the \emph{biadjacency matrix} of an
\emph{undirected bipartite graph $G_\dataset=(\dataset \cup \items, E)$}
corresponding to $\dataset$, where there is an edge\footnote{We always denote an
  edge of a bipartite graph corresponding to a dataset as $(a,b)$ with $a \in
  \dataset$ and $b \in \items$, i.e., as an element of $\dataset \times\items$,
  to make it clear which endpoint is a transaction and which is an item.} $(t,i)
\in E$ iff transaction $t$ contains the item $i$. 
\Cref{fig:transformation} (upper) depicts the bipartite graph corresponding to the dataset in the lower part of the figure. The left nodes (bottom nodes) model the three shopping baskets, while the right nodes (top nodes) represent the product bought.
Different matrices $M'$ and
$M''$ corresponding to $\dataset$ are the biadjacency matrices of bipartite
graphs that are \emph{structurally equivalent}, up to the labeling of the
transactions in $\dataset$.
In other words, all graphs corresponding to a dataset share the \emph{same
structural properties}, no matter their biadjacency matrices. To define our new
null model we use the graph $G_\odataset$.

\subsection{Preserving the Bipartite Joint Degree Matrix}\label{sec:model:bjdm}

One of our goals is to define a null model $\nullmodel = (\nullset, \nullprob)$
such that the datasets in $\nullset$ preserve not only the two fundamental
properties, but also an additional descriptive property of $\odataset$:
%Before formally defining $\nullset$, we must introduce the key concept of
the \emph{Bipartite Joint Degree Matrix (BJDM) $\bjdm{G_\odataset}$} of its
bipartite graph representation $G_\odataset$.

\begin{definition}[BJDM]\label{def:bjdm}
  Let $G = (L \cup R, E)$ be a bipartite graph, $k_L$ and $k_R$ be the largest degree 
  of a node in $L$ and $R$, respectively. 
  The \emph{Bipartite Joint Degree Matrix (BJDM) $\bjdm{G}$ of $G$}, is a $k_L
  \times k_R$ matrix whose $(i,j)$-th entry $\bjdm{G}[i,j]$ is the number of
  edges connecting a node $u \in L$ with degree $\degree{u} = i$ to a node $v
  \in R$ with degree $\degree{v} = j$, i.e.,
  \[
    \bjdm{G}[i,j] \doteq \card{\{(u,v) \in E \suchthat \degree{u} = i \wedge
    \degree{v} = j\}} \enspace.
  \]
\end{definition}

The BJDM of the graph in \cref{fig:transformation} (upper) is the following:

 \[
    \begin{pmatrix}
      0 & 0\\
      0 & 0\\
      0 & 0\\
      2 & 2\\
      2 & 3\\
      3 & 3\\
    \end{pmatrix}
\]

We define $\nullset$ as the set of all datasets $\dataset$ whose
transactions are built on $\items$ and whose corresponding bipartite
graph $G_\dataset$ has the same BJDM $\bjdm{G_{\dataset}}$.
We justify this choice by first showing that preserving the BJDM also preserves
the two fundamental properties, and then that it preserves additional ones.

\begin{fact}\label{fact:degree}
  For every $1 \le j \le k_R$, it holds
  \begin{equation}\label{eq:rightdegree}
    \card{\{ v \in R \suchthat \degree{v} = j\}} = \frac{1}{j} \sum_{i=1}^{k_L}
    \bjdm{G}[i,j],
  \end{equation}
  i.e., the BJDM $\bjdm{G}$ determines, for every $1 \le j \le k_R$, the number
  of vertices $v \in R$ of degree $\degree{v} = j$.

  Similarly, for every $1 \le i \le k_L$, it holds
  \begin{equation}\label{eq:leftdegree}
    \card{\{ u \in L \suchthat \degree{u} = i\}} = \frac{1}{i} \sum_{j=1}^{k_R}
    \bjdm{G}[i,j],
  \end{equation}
  i.e., the BJDM $\bjdm{G}$ determines, for every $1 \le i \le k_L$, the number
  of vertices $u \in L$ with degree $\degree{u} = i$.
\end{fact}

\begin{corollary}\label{corol:preserved}
  For any dataset $\dataset$, the BJDM $\bjdm{G_\dataset}$ determines, for every
  $1 \le j \le \card{\items}$, the number of transactions in $\dataset$ with length $j$.
  Also, it determines, for every $1 \le i \le \card{\dataset}$, the number of
  items with support $i$ in $\dataset$.
\end{corollary}

\Cref{corol:preserved} states that preserving the BJDM also preserves the two
fundamental properties. We now show an additional property that is preserved,
among others.

Let $\caternum{G_\odataset}$ be the number of \emph{simple paths of length
three} in $G_\odataset$, which, since $G_\odataset$ is bipartite, is also known
as the number of \emph{caterpillars} of
$G_\odataset$~\citep{aksoy2017measuring}. \Cref{corol:caternum} shows that
preserving the BJDM of $G_\odataset$ preserves the number of caterpillars. The
numbers of simple paths of length one and two are already preserved by
preserving the two fundamental properties, thus preserving also the number of
simple paths of length three is a natural step. Our desired result is a
corollary of \cref{lem:caternum}, which shows that $\caternum{G}$ can be
expressed through the BJDM.%

\begin{lemma}\label{lem:caternum}
  It holds
  \[
    \caternum{G} = \sum_{i=2}^{k_L}\sum_{j=2}^{k_R} \bjdm{G}[i,j](i-1)(j-1)
    \enspace.
  \]
\end{lemma}

\begin{proof}
  Each edge $(u,v) \in E$ is the middle edge of $(\degree{u} - 1) (\degree{v} -
  1)$ caterpillars, so
  \begin{equation}\label{eq:caternum}
    \caternum{G} =
    % 20220323 MR: Commenting out for now because it doesn't seem needed.
    %\sum_{\substack{u,v \in L \\ u \neq v}} \left(\degree{u} + \degree{v} - 2
    %\right) \cdot \card{\neigh{u, G} \cap \neigh{v, G}}\\
    %&=
    \sum_{(u,v) \in E} (\degree{u} - 1) (\degree{v} - 1)  \enspace.
  \end{equation}
  From here, we can conclude that
  \[
    \sum_{(u,v) \in E} (\degree{u} - 1) (\degree{v} -1) =
    \sum_{i=2}^{k_L}\sum_{j=2}^{k_R} \bjdm{G}[i,j](i-1)(j-1)
  \]
  because each edge $(u,v) \in E$ that connects a node $u \in L$ with degree
  $\degree{u} = i$ to a node $v \in R$ with degree $\degree{v} = j$ contributes
  $(i - 1) (j - 1)$ caterpillars to the summation in~\cref{eq:caternum}, and there are
  $\bjdm{G}[i,j]$ such edges.
\end{proof}

\begin{corollary}\label{corol:caternum}
  For any $\dataset$, the BJDM $\bjdm{G_\dataset}$ determines
  $\caternum{G_\dataset}$.
\end{corollary}

On the other hand, preserving the two fundamental properties and the number of
caterpillars is not sufficient to preserve the BJDM\@: as we now show, it is
easy to construct datasets that have the same transaction lengths, same item
supports, and same number of caterpillars as an observed dataset $\odataset$,
but whose BJDM is different than $\bjdm{G_\odataset}$. %
We show an example in \cref{fig:counter}. \rev{Both} bipartite graph\rev{s} in
\cref{fig:counter} \rev{have} three connected components \rev{each}, with a
total of 27 left-hand side nodes (light-blue, striped nodes) and 8 right-hand
side nodes (yellow, dotted nodes). It is easy to see that the two graphs have
the same degree distributions, and the same number of caterpillars (48). In the
upper graph, the leftmost component contains 36 caterpillars, while each of the
other two components contains 6 caterpillars, for a total of 48 caterpillars.
Similarly, in the lower graph, the leftmost component contains 36 caterpillars,
and the other two 6 caterpillars each. The two graphs have, nevertheless,
different BJDMs: in the upper graph there are edges connecting nodes with degree
4 to nodes with degree 5 (top left), but the lower graph has no such edge.

\begin{figure}[htb]
  \centering
  \includegraphics[width=\columnwidth]{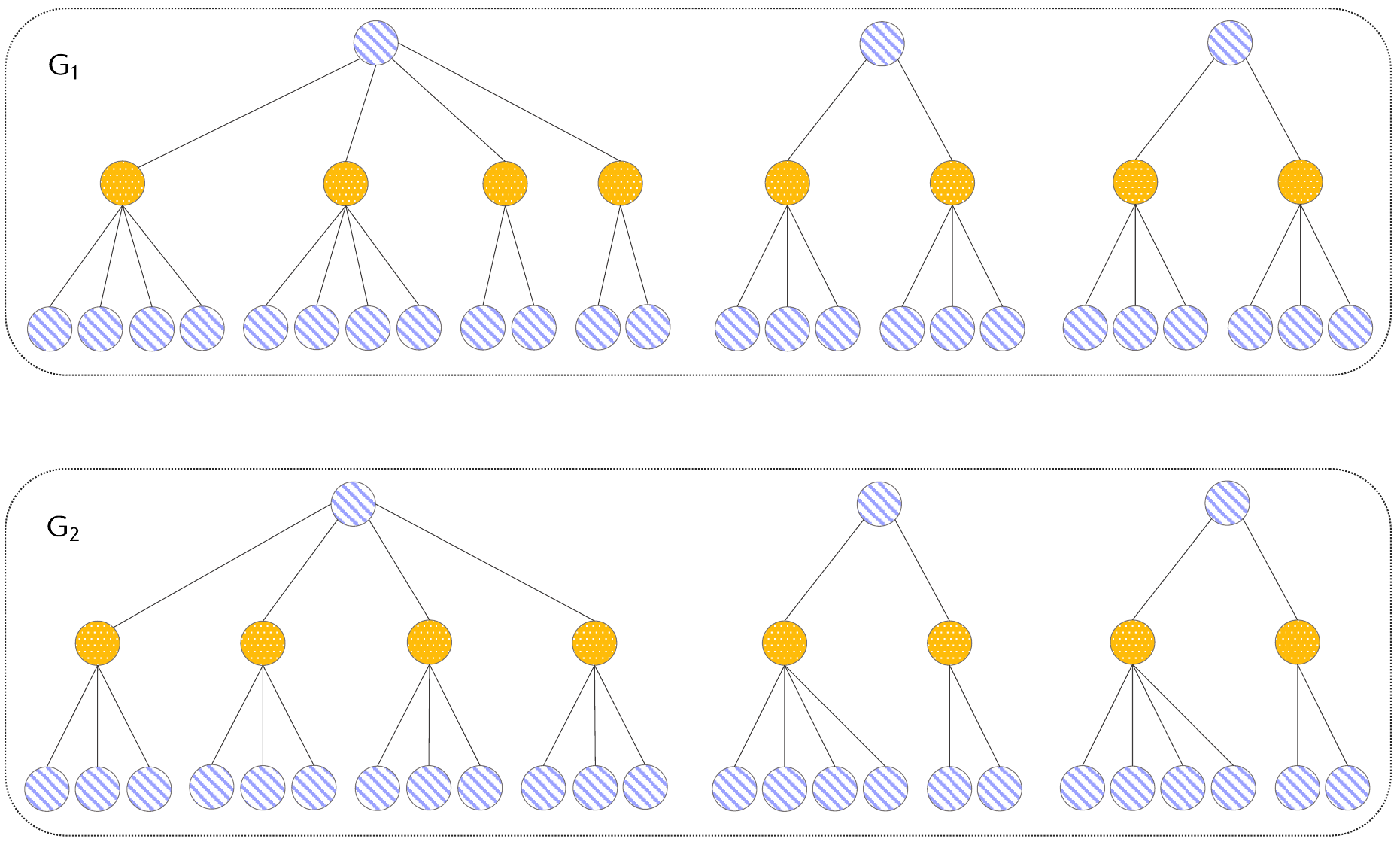}
  \caption{Two bipartite graphs with the same degree distributions and
  the same number of caterpillars, but different BJDMs.}%
  \label{fig:counter}
\end{figure}

%There is no immediate interpretation of the BJDM in terms of item(set)s % chktex 36
%and/or transactions, but the above result confirms that it captures additional
%structure of the observed dataset that is not preserved by the three basic
%properties.

We considered preserving more ``natural'' characteristics than the BJDM, such as
the support of each itemset of length two. However, doing so would lead to null
sets $\nullset$ that contain very few datasets in most cases, and are therefore
not very informative about the data generation process, as they are likely
overly constrained. Informally, the reason is that the biadjacency matrix
$M_\dataset$ of the graph $G_\dataset$ corresponding to any dataset $\dataset$
in such a $\nullset$ must satisfy $M_\dataset^{\phantom{T}} M_\dataset^\intercal = M_\odataset^{\phantom{T}}
M_\odataset^{\hspace{-.5pt}\intercal}$. %
% 20221221 MR: Commenting out because it is implied by the previous conditions
%and have the same row and column sums as $M_\odataset$.
Binary matrices $A$ and $B$ satisfying $A A^\intercal = B B^\intercal$ are known
as \emph{Gram mates}~\citep{Kirkland18,KimK22}. \Citet[Corol.\
1.1.1]{Kirkland18} shows an upper bound to the relative size of the set of Gram
mates w.r.t.\ the set of all binary matrices, which decreases as
the number of transactions in $\odataset$ and/or the number of items in $\items$
grow. While \citet{Kirkland18} and \citet{KimK22} construct infinite families of
Gram mates, they observe that these families ``possess a tremendous amount of
structure''~\citep[Sect.\ 4]{Kirkland18}, and it seems unlikely that such a
structure would ever occur on matrices corresponding to real datasets, to the
point that it is still an open question to determine whether a matrix $A$ even
admits \emph{any} Gram mate, which would at least allow us to determine whether
or not $\card{\nullset}=1$. On the other hand, if one can find at least one pair
of Gram mates, \citet[Sect.\ 5]{KimK22} give methods to build others (but
possibly not \emph{all}), thus if the open question is settled in a constructive
way, one may be able to sample from (a subset of) $\nullset$, if so
interested.

\rev{
Finally, we give an intuition about the properties that \algo preserves in addition to the fundamental ones.
Preserving the BJDM of a bipartite graph means preserving the number of edges connecting two nodes with given degrees.
This property implies, for instance, that the \emph{assortativity} of the graph~\citep{newman2002assortative}, i.e., the Pearson correlation coefficient of the vectors of degrees of nodes connected by an edge, is also maintained.
\Cref{fig:joint-distributions} shows an example of this property.
Assume to have a dataset with an empirical joint degree distribution as in \cref{fig:joint-alice}.
\algo preserves this joint degree distribution exactly.
Conversely, by preserving only the two fundamental properties, we only preserve the marginal distributions as in \cref{fig:joint-gmmt}.
In this latter case, the joint distribution is simply the product of the marginals, i.e., the marginals are assumed independent.
}

\begin{figure}[htb]
  \centering
  \begin{subfigure}[t]{0.49\textwidth}
    \includegraphics[width=\columnwidth]{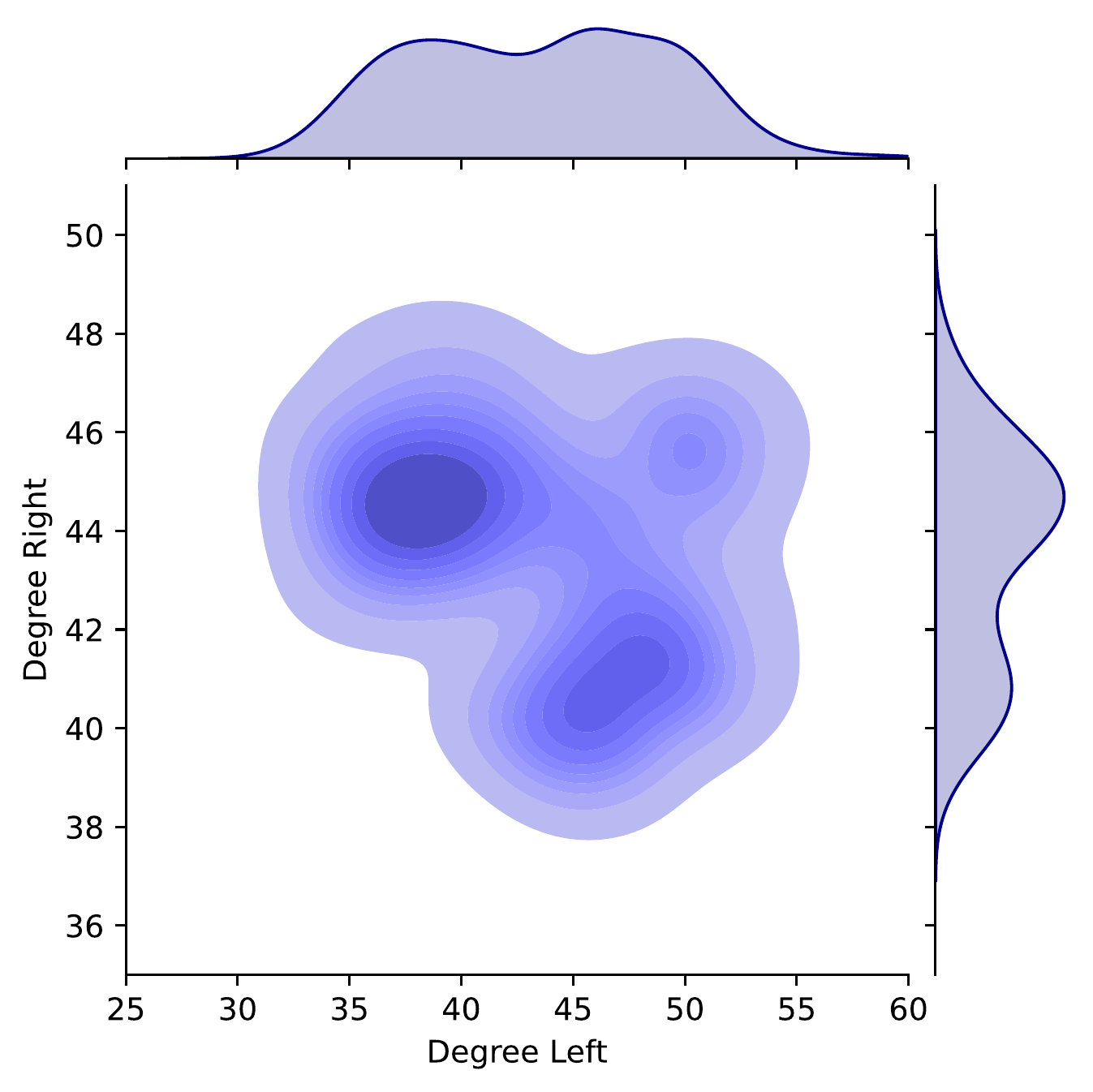}
    \caption{\rev{Joint distribution under \algo, which preserves the BJDM and maintains the degree assortativity of the dataset.}}
    \label{fig:joint-alice}
  \end{subfigure}  
  \begin{subfigure}[t]{0.49\textwidth}
    \includegraphics[width=\columnwidth]{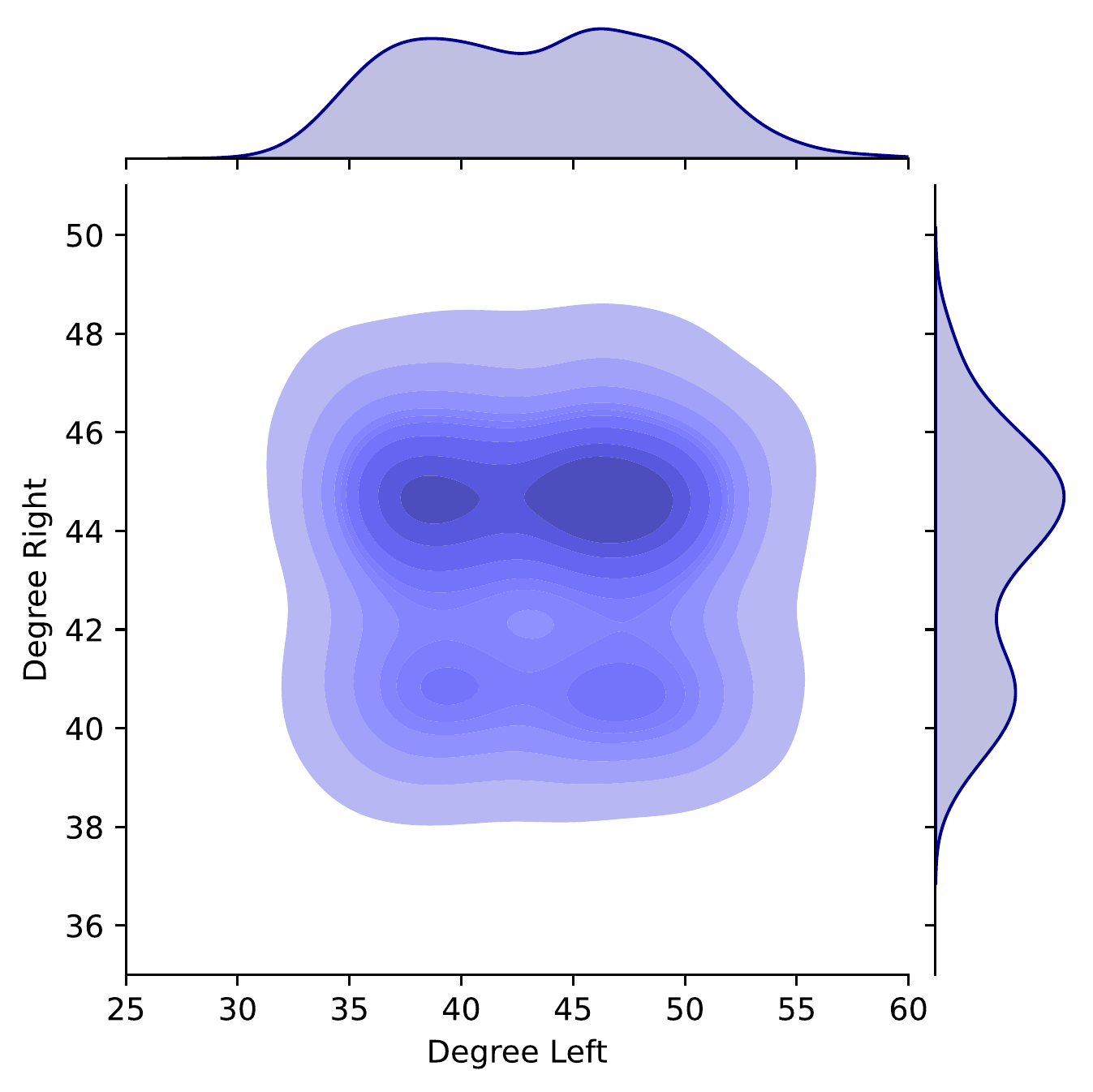}
    \caption{\rev{Joint degree distribution when preserving the two fundamental properties, where the left and right degree distributions are independent.}}
    \label{fig:joint-gmmt}
  \end{subfigure}
  \caption{\rev{Example of two different joint degree distributions of bipartite graphs with the same marginal degree distributions.}}
  \label{fig:joint-distributions}
\end{figure}

\section{Sampling from the Null Model}\label{sec:sampling}

We now present \algorso and \algocurve, two algorithms for sampling datasets
from the null model $\nullmodel = (\nullset, \nullprob)$.
%, i.e., for sampling elements of
%$\nullset$ according to the user-defined probability distribution $\nullprob$.

These algorithms take the MCMC approach with MH (see \cref{sec:prelims:mcmc}).
Their set of
states is the set $\matrices$ of matrices defined as follows. Fix $M_\odataset$
to be any of the biadjacency matrices of a bipartite graph corresponding to the
observed dataset $\odataset$. $\matrices$ contains all and only the matrices $M$
of size $\card{\odataset} \times \card{\items}$ such that, when considering $M$
as the biadjacency matrix of a bipartite graph $G_M$, it holds
$\bjdm{G_M} = \bjdm{G_\odataset}$.

$\matrices$ may contain multiple matrices associated to the same dataset (see
\cref{sec:model:matrices}), and different datasets may have a different number
of matrices in $\matrices$ associated to them. \algorso and \algocurve take this fact
into account to ensure that the sampling of datasets from $\nullset$ is done
according to $\nullprob$. For $M \in \matrices$, we use $\mattodat{M}$ to denote
the unique dataset corresponding to $M$, and for a dataset $\dataset \in
\nullset$, we use $\dattomat{\dataset}$ to denote the \emph{set} of matrices in
$\matrices$ corresponding to $\dataset$. \Citet[Lemma 3]{AbuissaLR23} give an
expression for the size $\copiesnum{\dataset } \doteq
\card{\dattomat{\dataset}}$ of $\dattomat{\dataset}$. The correctness of the two
algorithms relies on it so we report it here.

\begin{lemma}[\citealp{AbuissaLR23}, Lemma 3]\label{lem:numcopies}
  For any dataset $\dataset \in \nullset$, let $\{\ell_1,
  \dotsc, \ell_{z_\dataset}\}$ be the set of the $z_\dataset$ distinct lengths
  of the transactions in $\dataset$. For each $1 \le i \le z_\dataset$, let
  $T_i$ be the \emph{bag} of transactions of length $\ell_i$ in $\dataset$.
  Let $\bar{T}_i = \{\tau_{i,1}, \dotsc, \tau_{i,r_i}\}$ be the \emph{set} of
  transactions of length $\ell_i$ in $\dataset$, i.e., without duplicates. For
  each $1 \le j \le r_i$, let $Q_{i,j} \doteq \{ t' \in T_i \suchthat  t' =
  \tau_{i,j} \}$ be the \emph{bag} of transactions in $T_i$ equal to $\tau_{i,j}$ 
   (including $\tau_{i,j}$). Then, the number of matrices $M$ in
  $\matrices$ such that $\mattodat{M} = \dataset$ is
  \begin{equation}\label{eq:copiesnum}
    \copiesnum{\dataset} = \prod_{i=1}^{z_\dataset}
    \underbracket[0.187ex]{\binom{\card{T_i}}{\card{Q_{i,1}}, \dotsc,
    \card{Q_{i,r_i}}}}_{\text{multinomial coefficient}}
    = \prod_{i=1}^{z_\dataset} \frac{\card{T_i}!}{\prod_{j=1}^{r_i}
    \card{Q_{i,j}}!} \enspace.
  \end{equation}
\end{lemma}

\algorso and \algocurve take as inputs $\nullprob$ and the observed dataset $\odataset$.
It uses MH (see \cref{sec:prelims:mcmc}) to sample a matrix $M \in
\matrices$ according to a distribution $\statdistr$ (defined below), and returns
$\dataset = \mattodat{M} \in \nullset$ distributed according to $\nullprob$.
Both algorithms we present share the same set $\matrices$ of states, but they
have different neighborhood structures (i.e., the graphs used by MH for the two
algorithms have different sets of edges), different neighbor distributions
$\neighdistr{M}$, $M \in \matrices$, and different neighbor sampling procedures.

\subsection{\algorso: RSO-based Algorithm}\label{sec:sampling:swaps}

In our first algorithm, \algorso, the neighborhood structure over $\matrices$ is
defined using \emph{Restricted Swap Operations (RSOs)}~\citep[Sect.\
2]{czabarka2015realizations}.

\begin{definition}[Restricted Swap Operation (RSO)]\label{def:rso}
  Let $M$ be the $\card{L} \times \card{R}$ biadjacency matrix of a bipartite
  graph $G = (L \cup R, E)$.
  Let $1 \le a \neq b \le \card{L}$ and $1 \le c \neq d \le \card{R}$ be the
  indices of two rows and columns of $M$, respectively, such that
  \[
    M[a,c]=M[b,d] = 1  \wedge M[a,d]=M[b,c] = 0
  \]
  and such that \emph{at least one} of the following conditions holds
  \begin{align*}
    C_{ab} &= \text{``} \sum_{j=1}^{\card{R}} M[a,j] = \sum_{j=1}^{\card{R}}
    M[b,j] \text{''}\\
    C_{cd} &= \text{``} \sum_{i=1}^{\card{L}} M[i,c] = \sum_{i=1}^{\card{L}}
    M[i,d] \text{''} \enspace.
  \end{align*}
  The \emph{Restricted Swap Operation (RSO) $(a, c), (b, d) \rightarrow (a, d),
  (b, c)$ on $M$} is the operation that obtains the matrix $M'$ which is
  the same as $M$ but $M'[a,c] = M[a,d]$, $M'[a,d] = M[a,c]$, $M'[b,c] =
  M[b,d]$, and $M'[b,d] = M[b,c]$.
\end{definition}

\Cref{fig:rso} (left) depicts a bipartite graph, where dotted nodes indicate left nodes, and striped nodes indicate right nodes. For ease of presentation, we use different colors to denote nodes with the same degree.
A RSO in this graph is $(A, 1), (B, 5) \rightarrow (A, 5), (B, 1)$, because $A$
and $B$ satisfy condition $C_{ab}$ and the edges $(A, 5)$ and $(B, 1)$ are not
part of the graph. \Cref{fig:rso} (right) shows the graph resulting from the
application of the RSO\@. Dashed edges are edges involved in the RSO.%

\begin{figure}[htb]
  \centering
  \includegraphics[width=.7\columnwidth]{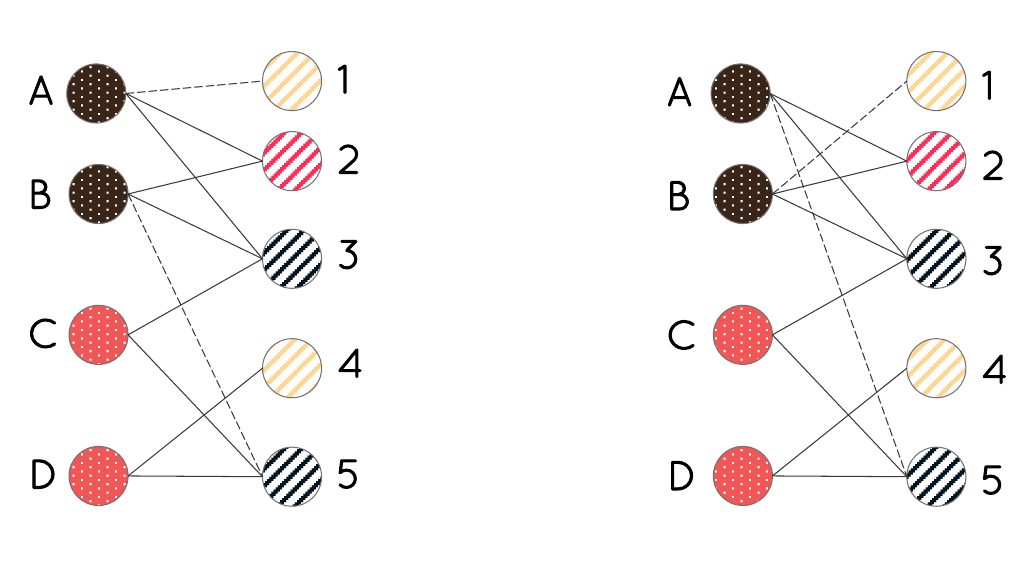}
  \caption{The RSO denoted with dashed edges transforms the left graph into the right graph. Different patterns denote nodes on different sides of the graph, while different colors denote different degrees.}%
  \label{fig:rso}
\end{figure}

Any RSO on $M \in \matrices$ results in a matrix $M'$ that
belongs to $\matrices$ as well. In the graph $G=(\matrices,
E)$ needed for MH, there is an edge from $M$ to $M'$
if there is a RSO from $M$ to $M'$. Additionally, there are \emph{self-loops}
from any $M \in \matrices$ to itself. These self-loops do not correspond to
RSOs, but they simplify the neighbor sampling procedure (described next). There
are zero or one RSOs between any pair of matrices in $\matrices$, but
$\matrices$ is strongly connected by RSOs~\citep[Thm.\
8]{czabarka2015realizations}.\footnote{The proof of \citep[Thm.\
8]{czabarka2015realizations} must be adapted, in a straightforward
way, to account for the fact that $\matrices$ contains biadjacency matrices of
bipartite graphs.}

RSOs are just one of the many possible operations that make $\nullset$
strongly connected. We discuss one such different operation in \cref{sec:sampling:curve}. %
% 20220609 MR: Commenting out for space saving. Restore for journal
%
Finding other operations to replace RSOs or to use in addition to RSOs is an
interesting research direction.

We now discuss the second ingredient needed to use MH\@: the distribution
$\neighdistr{M}$ over the set of neighbors $\neigh{M}$ of any $M \in
\matrices$. At first, using a distribution $\neighdistr{M}$ of the form
\[
  \neighdistr{M}(M') \doteq \begin{cases}
    \frac{2}{\card{\items}^2 \card{\dataset}^2} & M' \in \neigh{M} \setminus
    \{M\}\\
    1 - \frac{2(\card{\neigh{M}} - 1)}{\card{\items}^2 \card{\dataset}^2} &
    M' = M
  \end{cases}
\]
may seem an appealing option, because it could be realized by first drawing a
4-tuple $(a,b,c,d)$ uniformly at random from $\dataset \times \dataset \times
\items \times \items$, and then verifying whether $(a, c), (b, d)
\rightarrow (a, d), (b, c)$ is a RSO\@: if it is, one would set $M'$ to be the matrix resulting
from applying the RSO to $M$, otherwise $M' = M$. The major issue with this
approach is that, depending on $M$, the number of tuples that
must be drawn before finding one that is a RSO may be very large, thus slowing
down the process of moving on the graph. 
\rev{We briefly touch upon the convergence problem of this approach in \Cref{sec:exper}.}
Conversely, more complex
probability distributions that ensure drawing a neighbor different than $M$ are
quite easy to define, but come with the serious drawback that they need
expensive computation and bookkeeping of quantities such as $\card{\neigh{M}}$
and $\card{\neigh{M'}}$ for $M' \in \neigh{M}$ (due to~\cref{eq:mh}), or the
number of pairs of different rows or columns of the same lengths in $M$ and $M' \in
\neigh{M}$. The process of sampling a neighbor would then be much more
expensive, thus again slowing down the walk on the graph. We propose a
distribution over $\neigh{M}$ and a procedure to sample from it that strikes a
balance between statistical and computational ``efficiency'': the
probability of sampling $M$ is smaller than in the na\"{\i}ve case described
above, and sampling a neighbor is still quite efficient.

Let $M \in \matrices$ be the current state. For any $1 \le m \le \card{\items}$
(resp.~$1 \le n \le \card{\dataset}$), let $A_m$ be the set of row indices in $M$
whose rows have sum $m$ (resp.~let $B_n$ be set of column indices in $M$ whose
columns have sum $n$). To sample a neighbor $M'$ of $M$, we start by flipping a
fair coin. If the outcome is \emph{heads}, we first draw a row sum $1 \le m \le
\card{\items}$ with probability
\begin{equation}\label{eq:row-probability}
  \beta(m) = \slfrac{\binom{\card{A_m}}{2}}{\sum_{j=1}^{\card{\items}}
  \binom{\card{A_j}}{2}},
\end{equation}
and then we draw a pair $(a,b)$ of \emph{different} row indices in $A_m$
uniformly at random between such pairs. If the row of index $a$ and the row of
index $b$ in $M$ are identical, then we set $M' = M$. Otherwise, consider the
set $H_{a,b}$ of column index pairs $(p,q)$ such that
\[
  M[a,p] = M[b,q] \wedge M[a,q] = M[b,p] \wedge M[a,p] \neq M[a,q] \enspace.
\]
We draw a pair $(c,d)$ from $H_{a,b}$
uniformly at random. Then, either $(a,c),(b,d) \rightarrow (a,d),(b,c)$ or
$(a,d), (b,c) \rightarrow (a,c),(b,d)$ is a RSO by construction, and we set $M'$
to be the matrix obtained by performing this RSO on $M$. If the outcome of the
coin flip is \emph{tails}, we first draw a column sum $1 \le n \le
\card{\dataset}$ with probability
\begin{equation}\label{eq:col-probability}
  \gamma(n) = \slfrac{\binom{\card{B_n}}{2}}{\sum_{j=1}^{\card{\dataset}}
  \binom{\card{B_j}}{2}},
\end{equation}
and then we draw a pair $(c,d)$ of different column indices in $B_n$ uniformly
at random between such pairs. If the column of index $c$ and the column of index
$d$ in $M$ are identical, then we set $M' = M$. Otherwise, consider the set
$K_{c,d}$ of row index pairs $(p,q)$ such that
\[
  M[p,c] = M[q,d] \wedge M[p,d] = M[q,c] \wedge M[p,c] \neq M[p,d] \enspace.
\]
We draw a pair $(a,b)$ from $K_{c,d}$ uniformly at random. Then, either
$(a,c),(b,d) \rightarrow (a,d),(b,c)$ or is also a RSO by construction, and we
set $M'$ to be $(a,d),(b,c) \rightarrow (a,c),(b,d)$ is a RSO by construction,
and we set $M'$ to be the matrix obtained by performing this RSO on $M$.

This procedure induces a probability distribution $\neighdistr{M}$ over
$\neigh{M}$. Let us analyze $\neighdistr{M}(M')$ for $M' \neq M$. W.l.o.g., let
$(a,c),(b,d) \rightarrow (a,d),(b,c)$ be the sampled RSO, and let $M'$ be the neighbor of
$M$ obtained by performing such RSO on $M$\@. Recall that the sampled RSO is the
only RSO from $M$ to $M'$. Consider the following events:
\begin{align*}
  &E_{\mathrm{row}} \doteq \text{``rows $a$ and $b$ of $M$ have the same
  row sum $m$'';}\\
  &E_{\mathrm{col}} \doteq \text{``columns $c$ and $d$ of $M$ have the same
  column sum $n$''.}
\end{align*}
There are three possible cases for the probability $\neighdistr{M}(M')$ of
sampling $M'$:
\begin{itemize}
  \item if only $E_{\mathrm{row}}$ holds, then
    \begin{equation}\label{eq:neighdistrrow}
      \neighdistr{M}(M') = \frac{1}{2}
      \frac{1}{\sum_{i=1}^{\card{\items}} \binom{\card{R_i}}{2}}
      \frac{1}{\card{H_{a,b}}};
    \end{equation}
  \item if only $E_{\mathrm{col}}$ holds, then
    \begin{equation}\label{eq:neighdistrcol}
      \neighdistr{M}(M') = \frac{1}{2}
      \frac{1}{\sum_{j=1}^{\card{\dataset}} \binom{\card{C_j}}{2}}
      \frac{1}{\card{K_{a,b}}};
    \end{equation}
  \item if both $E_{\mathrm{row}}$ and $E_{\mathrm{col}}$ hold, then $M$' (i.e.,
    the RSO) may be sampled regardless of the outcome of the coin flip. Thus,
    $\neighdistr{M}(M')$ is the sum of r.h.s.'s of~\cref{eq:neighdistrrow}
    and~\cref{eq:neighdistrcol}.
\end{itemize}
We do not need to analyze $\neighdistr{M}(M)$ because if $M$ is drawn as the
``neighbor'', then MH will definitively select $M$ as the next state, thus we do
not need to explicitly compute its probability.

It holds that $\neighdistr{M}(M') = \neighdistr{M'}(M)$, which greatly
simplifies the use of MH\@: from~\cref{eq:mh}, we see that, thanks to the
construction of the graph and the definition of the neighbor sampling
distribution, we really only need the distribution $\statdistr$ over
$\matrices$. We define it as
\begin{equation}\label{eq:statdistrrso}
  \statdistr(M) = \frac{\nullprob(\mattodat{M})}{\copiesnum{\mattodat{M}}},
\end{equation}
where $\copiesnum{\mattodat{M}}$ is from~\cref{eq:copiesnum}. The following
lemma shows that \algorso samples a dataset $\dataset$ from
$\nullset$ according to $\nullprob$, i.e., it samples from the null model.

\begin{lemma}\label{lem:swapcorrectness}
  Let $\dataset \in \nullset$. \algorso outputs $\dataset$ with probability $\nullprob(\dataset)$.
\end{lemma}

\begin{proof}
  Let $M \in \matrices$. From the correctness of MH we have that \algorso
  samples $M$ according to $\statdistr$ from~\cref{eq:statdistrrso}. The thesis
  then follows from noticing that $\dataset$ is returned in output whenever
  \algorso samples one of the $\copiesnum{\dataset}$ matrices in $\matrices$
  corresponding to $\dataset$.
\end{proof}

\rev{\Cref{alg:alice} illustrates the main steps performed by \algorso to sample a dataset in $\nullset$.
The algorithm receives in input a matrix $M \in \matrices$ and a number of swaps $s$ sufficiently large for convergence.
Previous works estimated that a number of steps in order of the number of $1$s in $M$ is sufficient. We will discuss this aspect in \Cref{sec:exper}.}
\begin{algorithm}[thb]
    \caption{\rev{\algo}}\label{alg:alice}
    \begin{algorithmic}[1]
    \Require Matrix $M \in \matrices$, Number of Swaps $s$
    \Ensure Dataset $\dataset$ sampled from $\nullset$ with probability $\nullprob(\dataset)$ 
    \State $\copiesnum{\mattodat{M}} \gets $\Cref{eq:copiesnum}
    \State $i \gets 0$
    \While{$i < s$}
      \State $i \gets i + 1$
      \State $\mathsf{out} \gets$ flip a fair coin
      \If{$\mathsf{out}$ is \emph{heads}}
        \State $a$, $b$ $\gets$ different row indices drawn u.a.r. such that $C_{ab}$ holds 
        \State $c$, $d$ $\gets$ pair drawn u.a.r. from $H_{ab}$\label{line:rsoa}
      \Else
        \State $c$, $d$ $\gets$ different column indices drawn u.a.r. such that $C_{cd}$ holds
        \State $a$, $b$ $\gets$ pair drawn u.a.r. from $K_{cd}$\label{line:rsob}
      \EndIf
      \State $M' \gets $ perform $(a, c), (b, d) \rightarrow (a, d), (b, c)$ on $M$\label{line:apply}
      \State $\copiesnum{\mattodat{M'}} \gets $\Cref{eq:copiesnum}
      \State $p \gets $ random real number in $[0,1]$
      \State $a \gets \min\left(1, \copiesnum{\dataset}/\copiesnum{\dataset'}\right)$
      \If{$p \leq a$}
        $M \gets M'$
      \EndIf
    \EndWhile
    \State \Return $\mattodat{M}$
    \end{algorithmic}
\end{algorithm}

%\question[from=Matteo,date=4/6]{Shall we say something about the mixing time?
%  \citet{amanatidis2022rapid} show that a \emph{different} Markov chain
%  (similar to, but not the same as, the one based on the rejection sampling
%  procedure), is rapidly mixing, but we cannot really conclude anything about
%  that. Perhaps when we measure the mixing time experimentally, we can cite
%  that paper with a comment.}

\subsection{\algocurve: Adapting Curveball}\label{sec:sampling:curve}
% 20220405 MR: commenting out because it is something that we have to keep in
% mind, so I prefer not to delete it.
%
%\mynote[from=Giulia]{Regarding the Curveball algorithm, some cite this work for
%bipartite networks\cite{strona2014fast}. This one
%\url{https://arxiv.org/pdf/2112.04017.pdf} has not been published yet, hence I
%am not sure whether we should mention it or not.}
%
%\reply[from=Matteo]{Yes, the corrigendum of \citep{strona2014fast}
%  (\url{https://www.nature.com/articles/ncomms13086.pdf}) makes it
%  clear that they essentially ``forgot'' to cite \citet{verhelst2008efficient}.
%
%  Note that any binary matrix with $m$ rows and $n$ columns can be seen as the
%  biadjacency matrix of a bipartite graph.
%
%  I had seen the arxiv one before. To me, it seems mostly an implementation
%  detail, the algorithm is really the same (\citet{verhelst2008efficient} is a
%  mathematician, so he doesn't discuss algorithmic details).
%}
We now introduce a second algorithm, \algocurve, that can essentially perform
multiple RSOs at each step of the Markov chain, thus leading to a faster mixing of
the chain, i.e., to fewer steps needed to sample a dataset from $\nullmodel$.
Our approach adapts the \textsc{CurveBall}
algorithm~\citep{verhelst2008efficient}, which samples a matrix from the space
of binary matrices with fixed row and column sums, to use RSOs. \algocurve is
also an MCMC algorithm that uses MH\@. The vertex set of the graph $G =
(\matrices, E)$ is still the set $\matrices$ previously defined, but \algocurve
uses a different set of edges than \algorso: there is an edge $(M,M') \in E$
from a matrix $M \in \matrices$ to $M' \in \matrices$ iff $M'=M$ or there is a
\emph{Restricted Binomial Swap Operation (RBSO)} on $M$ that results in $M'$.
RBSOs are defined as follows.

\begin{definition}[Restricted Binomial Swap Operation (RBSO)]\label{def:rbso}
  Given a matrix $M \in \matrices$, let $a$ and $b$ be the indices of two
  \emph{distinct and different} rows of $M$ with the same row sum. Let
  $Z_a(M,b)$ be the set of column-indices $q$ such that $M[a,q]=1$ and
  $M[b,q]=0$, and define $Z_b(M,a)$ similarly (it holds $Z_a(M,b) \cap Z_b(M,a) =
  \emptyset$ and $\card{Z_a(M,b)} = \card{Z_b(M,a)}$). Let $U$ be any subset of
  $Z_a(M,b) \cup Z_b(M,a)$ of size $\card{Z_a(M,b)}$.
  The \emph{row Restricted Binomial Swap Operation (rRBSO) $(a, b, U)$ on
  $M$} is the operation that obtains a matrix $M'$ such that $M'[i,j] = M[i,j]$
  except for $i \in \{a,b\}$, and such that the rows of index $a$ and $b$ of
  $M'$ are
  \[
    M'[a,q] \doteq
    \begin{cases}
      M[a,q] & q \notin Z_a(M,b) \cup Z_b(M,a)\\
      1 & q \in U\\
      0 & q \in (Z_a(M,b) \cup Z_b(M,a)) \setminus U
    \end{cases}
  \]
  and
  \[
    M'[b,q] \doteq
    \begin{cases}
      M[b,q] & q \notin Z_a(M,b) \cup Z_b(M,a)\\
      0 & q \in U\\
      1 & q \in (Z_a(M,b) \cup Z_b(M,a)) \setminus U
    \end{cases}
  \]
  A corresponding definition for a \emph{column RBSO (cRBSO)} can be given for
  $a$ and $b$ being the indices of two distinct and different columns with the
  same column sum.

  We use ``RBSO'' to refer to either a rRBSO or a cRBSO, and the set of RBSOs is
  composed by all rRBSOs and cRBSOs.
\end{definition}

\Cref{fig:rbso} (left) depicts a bipartite graph using the same style used in \cref{fig:rso}. 
Let $a = 1$ and $b = 2$, which are two right nodes with the same degree but different sets of neighbors. 
Then, $Z_a(M,b) = \{A,D\}$ and $Z_b(M,a) = \{B,G\}$.
For $U = \{B, G\}$, the RBSO $(a,b,U)$ generates the graph in \cref{fig:rbso}
(right). Dashed edges are edges involved in the RBSO.%

\begin{figure}[htb]
  \centering
  \includegraphics[width=.7\columnwidth]{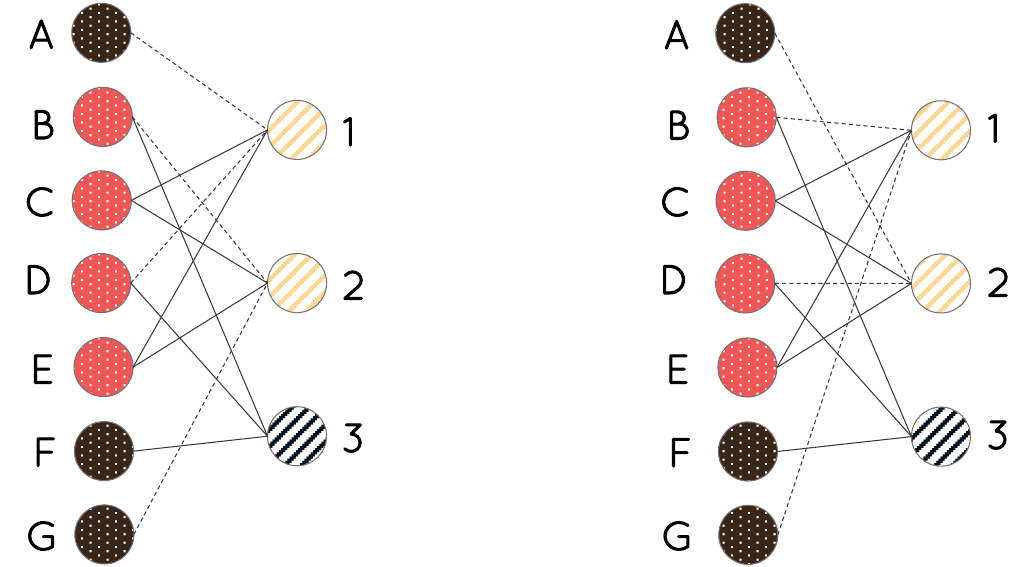}
  \caption{The RBSO denoted with dashed edges transforms the left graph into the right graph. Different patterns denote nodes on different sides of the graph, while different colors denote different degrees.}%
  \label{fig:rbso}
\end{figure}

Any RBSO on a matrix $M$ preserves $\bjdm{M}$, and any RBSO can be seen as a
sequence of RSOs. For any RSO $(a,c), (b,d) \rightarrow (a,d), (b,c)$ on $M$
there is an equivalent RBSO $(a, b, (Z_a(M,b) \setminus \{c\}) \cup \{ d\})$ from
$M$, and thus the graph $G = (\matrices, E)$ is also strongly connected, as it has
all the edges which are created by RSOs, plus potentially others.

\begin{fact}\label{fact:inverserbso}
  Let $(a, b, U)$ be a cRBSO (resp.~rRBSO) from $M$ to $M' \in \neigh{M}$ with $M'
  \neq M$. Then $(a, b, Z_a(M,b))$ is a cRBSO (resp.~rRBSO) from $M'$ to $M$.
\end{fact}

\begin{lemma}\label{lem:rbsonumber}
  There are either one or two RBSOs from $M \in \matrices$ to $M' \in \neigh{M}$ with
  $M' \neq M$. When there are \emph{two} RBSOs, one is a cRBSO and the other is
  a rRBSO.%
\end{lemma}

\begin{proof}
  Let us start from the second part of the thesis. If $(a, b, \{c\})$ is a
  cRBSO (resp.~rRBSO) from $M$ to $M'$, then
  \[
    (c, (Z_a(M,b) \cup Z_b(M,a)) \setminus \{c\}, \{a\})
  \]
  is a rRBSO (resp.~cRBSO) from $M$ to $M'$.

  The fact that there can only be one or two RBSOs is a consequence of
  \cref{fact:inverserbso}.
\end{proof}

In order for two RBSOs from $M$ to $M'$ to exist, it is necessary that
$\card{Z_a(M,b)} = \card{Z_b(M,a)} = 1$, the columns at indices $a$ and $b$
have the same sum, and the rows at indices $c$ and $(Z_a(M,b) \cup Z_b(M,a)) \setminus
\{c\}$ have the same sum.

\begin{corollary}\label{corol:samenumber}
  For any two $M$ and $M'$, there is the same number of RBSOs from $M$ to $M'$
  as from $M'$ to $M$.
\end{corollary}

Let us now give the procedure to sample a neighbor $M' \in \neigh{M}$ of $M$.
The procedure is similar to the one for \algorso. First, we flip a fair coin. If
the outcome is \emph{heads}, we draw a row sum $1 \le m \le \card{\items}$
with probability
as per \cref{eq:row-probability},
%\[
%  \slfrac{\binom{\card{R_m}}{2}}{\sum_{j=1}^{\card{\items}}\binom{\card{R_j}}{2}},
%\]
and then we draw a pair $(a,b)$ of different row indices in $R_m$ uniformly at
random between such pairs. If the row of index $a$ and the row of index $b$ in
$M$ are identical, then we set $M' = M$. Otherwise, we compute the set $Z_a(M,b) \cup
Z_b(M,a)$ defined in \cref{def:rbso} and the cardinality $\card{Z_a(M,b)}$ with a linear
scan of the rows $a$ and $b$. By using reservoir sampling~\citep{Vitter85}, we obtain $U$ through a
linear scan of $Z_a(M,b) \cup Z_b(M,a)$. If the outcome of the coin flip is \emph{tails},
we first draw a column sum $1 \le n \le \card{\dataset}$ with probability
as per \cref{eq:col-probability},
%\[
%  \slfrac{\binom{\card{C_n}}{2}}{\sum_{j=1}^{\card{\dataset}}\binom{\card{C_j}}{2}},
%\]
then we draw a pair $(a,b)$ of different column indices in $C_n$ uniformly at
random between such pairs. We then proceed in a fashion similar as for the row
case. The purpose of flipping the coin at the start is to ensure that we can
sample both rRBSOs (when the outcome is heads), and cRBSOs (otherwise).

The probability $\neighdistr{M}(M')$ of sampling
a RBSO $(a, b, U)$ on $M$ that results in $M'$, is not uniform. Rather than
giving the expression for it, we use the fact that, in order to use MH, we
really only need the distribution $\statdistr$ over $\matrices$, and the
\emph{ratio} $\sfrac{\neighdistr{M'}(M)}{\neighdistr{M}(M')}$
(see~\cref{eq:mh}), and we now show that $\neighdistr{M}(M') =
\neighdistr{M'}(M)$, i.e., the ratio is always 1.

\begin{lemma}
  Let $M \in \matrices$ and $M' \in \neigh{M}$. Then $\neighdistr{M}(M') =
  \neighdistr{M'}(M)$.
\end{lemma}

\begin{proof}
  We assume that $M' \neq M$, otherwise the thesis is obviously true. For ease
  of presentation, we focus on the case where there is only a cRBSO $(a, b, U)$
  from $M$ to $M'$. The analysis for the case when there is only a rRBSO follows
  the same steps, and the one for the case when there is both a cRBSO and a
  rRBSO follows by combining the two cases.

  From \cref{fact:inverserbso}, the cRBSO $(a, b, Z_a(M,b))$ goes from $M'$ to $M$.
  The probability that the coin flip is tails is the same no matter whether the
  current state is $M$ or if it is $M$, as is the probability, given that the
  outcome was tails, of sampling the columns indices $a$ and $b$. By definition,
  it holds that $\card{U} =  \card{Z_a(M,b)}$, and it is easy to see that
  $Z_a(M,b) \cup Z_b(M,a) = Z_a(M',b) \cup Z_b(M',a)$, thus the probability of
  sampling $U$ when the current state is $M$ and we have sampled $a$ and $b$,
  and the probability of sampling $Z_a(M,b)$ when the current state is $M'$ and
  we have sampled $a$ and $b$ are the same. Thus, the probability of sampling
  $(a, b, U)$ when the current state is $M$ is the same as the probability of
  sampling $(a, b, Z_a(M,b))$ when the current state is $M'$, and the proof is
  complete.
\end{proof}

Thus, to use MH, we really only need the distribution $\statdistr$ over
$\matrices$. As in \cref{sec:sampling:swaps}, in order to sample a dataset $D
\in \nullset$ according to $\nullprob$, we want to sample a matrix $M \in
\matrices$ with the probability given in~\cref{eq:statdistrrso}. We thus have
all the ingredients to use MH, and our description of \algocurve is complete.
\rev{Note that \algocurve follows the same structure presented in \Cref{alg:alice} but samples a rRBSO $(a,b,U)$ at line \ref{line:rsoa}:}
\begin{algorithm}[hb]
    \begin{algorithmic}[1]
    \small
    \Statex \rev{$U \subset Z_a(M,b) \cup Z_b(M,a)$ s.t. $\card{U} = \card{Z_a(M,b)}$ obtained via reservoir sampling}
    \end{algorithmic}
\end{algorithm}

\rev{and a cRBSO $(c,d,U)$ at line \ref{line:rsob}:}
\begin{algorithm}[hb]
    \begin{algorithmic}[1]
    \small
    \Statex \rev{$U \subset Z_c(M,d) \cup Z_d(M,c)$ s.t. $\card{U} = \card{Z_c(M,d)}$ obtained via reservoir sampling}
    \end{algorithmic}
\end{algorithm}
% !TEX root = ../main-kais.tex
\section{Sequence Datasets}\label{sec:seq}

Previous work studied null models for testing the statistical
significance of results obtained from other kinds of datasets, such as
sequence datasets~\citep{TononV19,PinxterenC21,JenkinsWGR22,LowKRKP13}. We now
define a new null model for sequence datasets to also preserve the
BJDM, and we introduce a new algorithm \algoseq to sample from this null model.

\subsection{Preliminaries on sequence datasets and multi-graphs}\label{sec:seq:prelims}

Let us start with a brief description of sequence datasets and related concepts.
A \emph{sequence} is a finite \emph{ordered list} (or a \emph{vector}) of
not-necessarily-distinct itemsets, i.e., $S = \seq{A_1, \dotsc, A_\ell}$ for
some $\ell \geq 1$, with $A_i \subseteq \items$, $1 \leq i \leq \ell$. Itemsets
$A_i$ \emph{participate} in $S$, and we denote this fact with $A_i \in S$, $1
\leq i \leq \ell$. The \emph{length} $\card{S}$ of a sequence is the number of
itemsets participating in it. A \emph{sequence dataset} $\dataset$ is a finite
bag of sequences, which, as elements of $\dataset$, are known as
\emph{seq-transactions}. The \emph{support} $\supp{\dataset}{A}$ of an itemset
$A$ in $\dataset$ is the number of seq-transactions of $\dataset$ in which $A$
participates. The \emph{multi-support} $\msupp{\dataset}{A}$ of $A$ in
$\dataset$ is the number of times that $A$ participates \emph{in total} in the
seq-transactions of $\dataset$. For example, in the dataset $\dataset = \{
\seq{A,B}, \seq{A,C,A}, \seq{B,C} \}$, it holds that $\supp{\dataset}{A} = 2$ and
$\msupp{\dataset}{A} = 3$.

\begin{figure}
  \centering
  \includegraphics[width=.75\textwidth]{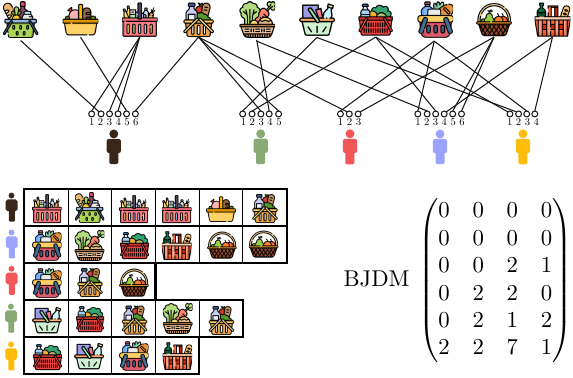}
  \caption{Example of sequence dataset (lower left), corresponding multi-graph
    (top), and BJDM of the multi-graph (lower right).}\label{fig:seqex}
\end{figure}

A sequence dataset $\dataset$ can be represented as a bipartite \emph{multi-}graph
$G_\dataset = (L \cup R, E)$, where $L$ are the seq-transactions of
$\dataset$, and $R$ is the \emph{set} of all and only the
itemsets with support at least $1$ in $\dataset$, i.e., participating in at
least one seq-transaction of $\dataset$. Each vertex $v \in L$ has
degree\footnote{In multi-graphs, the degree of a vertex $v$ is still the number
of edges incident to it, so each edge is counted, even if multiple edges connect
$v$ to the same vertex.} equal to the length of the corresponding
seq-transaction $S_v$ of $\dataset$, i.e., $\degree{v} = \card{S_v}$. Each
vertex $v \in L$ has $\degree{v}$ \emph{ports}, which can be thought as the
``locations'' where the edges ``connect'' to $v$. The ports are arbitrarily
labeled from $1$ to $\degree{v}$. This labeling is needed to define the edge
\emph{multi-}set $E$ as follows: there is an edge between $v \in L$ and $w \in
R$ using port $k$ of $v$ iff the itemset $B_w$ corresponding to the vertex $w$
appears in position $k$ of $S_v$, i.e., iff $S_v = \seq{A_1, \dotsc, A_{k-1},
B_w, A_{k+1}, \dotsc, A_{\card{S_v}}}$. We denote this edge as $(v,k,w)$, thus
$E$ can also be thought as a set of such tuples. To the best of our knowledge,
the one we just gave is the first description of sequence datasets as bipartite
multi-graphs, which is somewhat surprising because representing transactional
datasets as bipartite graphs has been a standard practice for a long time.

The definition of BJDM from \cref{def:bjdm} is also valid for multi-graphs.
\Cref{fig:seqex} shows an example of a sequence dataset (lower left), the
corresponding multi-graph (top), and its BJDM (lower right).

\subsection{BJDM-preserving null model for sequence datasets}\label{sec:seq:null}

\Citet{TononV19} introduce a null model $\nullmodel = (\nullset, \nullprob)$ for
sequence datasets that can be seen as an adaptation of \citet{GionisMMT07}'s
null model for transactional datasets. It preserves the following two properties
of an observed dataset $\odataset$:

\begin{itemize}
  \item the distribution of the seq-transaction lengths, i.e., for any
    seq-transaction length $\ell \in [1, \max_{S \in \odataset} \card{S}]$, any
    $\dataset \in \nullset$ contains the same number of transactions of length
    $\ell$ as $\odataset$; and
  \item the multi-support of the itemsets participating in the
    seq-transactions of $\odataset$, i.e., for any $A \subseteq \items$ and
    $\dataset \in \nullset$, $\msupp{\odataset}{A} =
    \msupp{\dataset}{A}$.
\end{itemize}

It should be evident how these two properties can be mapped to the two
fundamental properties defined in \cref{sec:prelims:hyptest} for transactional
datasets, with the difference that itemsets participating in seq-transactions
play the role that was of items in transactional datasets. \Citet{TononV19} gave
a MCMC algorithm to sample from this null model, while \citet{JenkinsWGR22} gave
an exact sampling algorithm.

The null model we define for sequence datasets preserves the BJDM of the
multi-graph corresponding to the observed dataset.
The following property can be derived in a way similar to that from
  \cref{corol:preserved}, and confirms that preserving the BJDM
  also preserves the two above properties.

\begin{corollary}\label{corol:preserved-seq}
  For any sequence dataset $\dataset$, the BJDM $\bjdm{G_\dataset}$ determines,
  for every $1 \le j \le \max_{S \in \dataset} \card{S}$, the number of
  seq-transactions in $\dataset$ with length $j$. Also, it determines, for every
  $1 \le i \le \card{\dataset}$, the number of itemsets with multi-support $i$ in
$\dataset$.
\end{corollary}

On the other hand, it is not true that preserving the BJDM also preserves the
number of caterpillars on multi-graphs, i.e., there is no equivalent of
\cref{lem:caternum,corol:caternum}. The reason is that the BJDM does not encode
information that allows the distinction between simple and multiple edges, i.e.,
the fact that a vertex with degree $x$ may have any number of neighbors between
1 and $x$. It is also easy to come up with examples showing that it is not true
that preserving the BJDM on multi-graphs preserves the number of
\emph{not-necessarily simple} paths of length three composed of three distinct
edges.
For instance, the multi-graph in \cref{fig:mrso-counter} (left) includes the following $10$ paths of length three: $(\beta,3,D)-(\beta,1,B)-(\alpha,3,B)$, $(\beta,3,D)-(\beta,1,B)-(\alpha,2,B)$, $(\beta,2,C)-(\beta,1,B)-(\alpha,3,B)$, $(\beta,2,C)-(\beta,1,B)-(\alpha,2,B)$, $(\beta,1,B)-(\alpha,3,B)-(\alpha,2,B)$, $(\beta,1,B)-(\alpha,2,B)-(\alpha,3,B)$, $(\beta,1,B)-(\alpha,3,B)-(\alpha,1,A)$, $(\beta,1,B)-(\alpha,2,B)-(\alpha,1,A)$, $(\alpha,3,B)-(\alpha,2,B)-(\alpha,1,A)$, and $(\alpha,2,B)-(\alpha,3,B)-(\alpha,1,A)$.
The multi-graph on the right, which can be obtained by applying the mRSO $(\alpha,1,A), (\beta,1,B) \rightarrow (\alpha,1,B), (\beta,1,A)$ has the same BJDM but only six paths of length three: $(\alpha,1,B)-(\alpha,2,B)-(\alpha,3,B)$, $(\alpha,1,B)-(\alpha,3,B)-(\alpha,2,B)$, $(\alpha,2,B)-(\alpha,1,B)-(\alpha,3,B)$, $(\alpha,2,B)-(\alpha,3,B)-(\alpha,1,B)$, $(\alpha,3,B)-(\alpha,2,B)-(\alpha,1,B)$, and $(\alpha,3,B)-(\alpha,1,B)-(\alpha,2,B)$.

\begin{figure}[htb]
  \centering
  \includegraphics[width=\textwidth, trim={0 60 0 60}, clip]{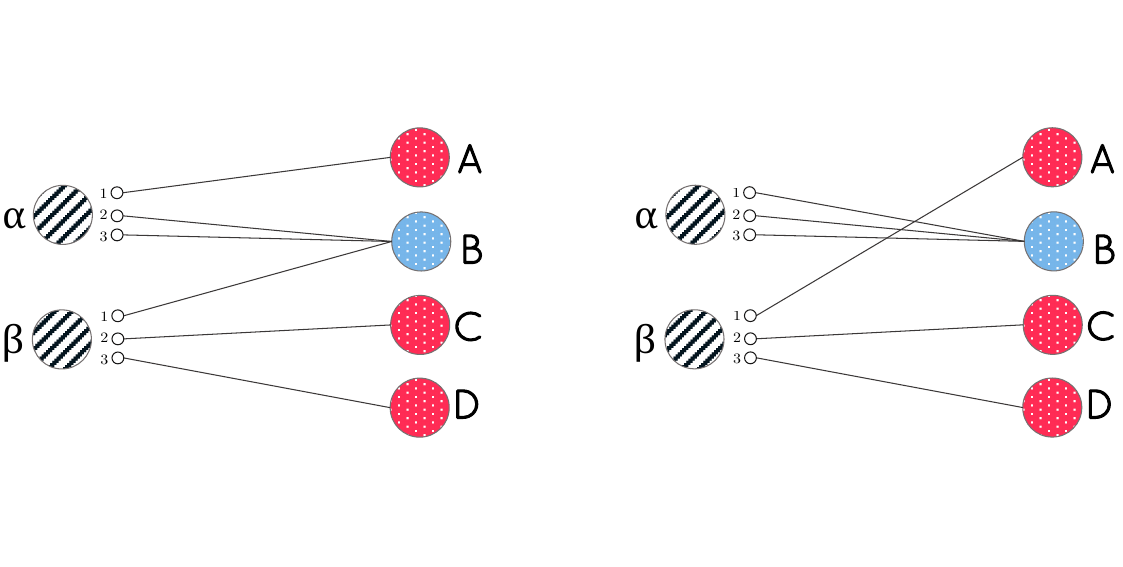}
  \caption{Two bipartite multi-graphs with the same BJDM but different numbers of paths of length three. Different patterns denote nodes on different sides of the multi-graph, while different colors denote different degrees.}
  \label{fig:mrso-counter}
\end{figure}

Nevertheless, since the multi-graph corresponding to a sequence dataset
may actually be a simple graph, preserving the BJDM preserves more structure of
the observed dataset than just the two fundamental properties, as we discussed
for the counterexample from \cref{fig:counter}.

\subsection{\algoseq: \algo for sequence datasets}\label{sec:seq:algo}

We now discuss \algoseq, our algorithm for sampling from the BJDM-preserving
null model for sequence dataset, which was defined in the previous section. Like
the other members of the \algo family, \algoseq also takes the MCMC approach
with MH\@. Its set of states though, is no longer the set $\matrices$ of
biadjacency matrices, but a set $\graphs$ of bipartite multi-graphs defined as follows.
Given the observed sequence dataset $\odataset$, let $G_{\odataset}=(L \cup R,
E)$ be the multi-graph corresponding to it. $\graphs$ contains all and only the
bipartite multi-graphs with node sets $L$ and $R$, and with the same BJDM as
$G_{\odataset}$. We remark that $\graphs$ therefore includes also bipartite
multi-graphs that are isomorphic to each other but differ for the ports to which
the edges are connected, as such graphs represent different sequence
datasets where the order of the itemsets in (some of) the sequences is shuffled.

The reason for not using the set $\matrices$ of biadjacency matrices as the
state space of \algoseq is that a biadjacency matrix does not capture the
entirety of the structure of a multi-graph corresponding to a sequence
dataset, as it does not encode the information about the ports. It is important
to understand that \algorso could have been easily presented in
\cref{sec:sampling:swaps} with a state space composed of graphs, rather than
biadjacency matrices. We chose not to do that because the presentation of
\algocurve greatly benefits from using matrices, although even in this case we
could have used graphs, given that in the simple graph case, there is a
bijection between bipartite graphs and biadjacency matrices. The flow and the
notation in the following presentation of \algoseq are similar to the one for
\algorso, to highlight the many similarities between the two algorithms, but
there are also many crucial differences.

We now define the concept of \emph{multi-graph Restricted Swap Operation (mRSO)}
as an operation that is applied to a multi-graph $G$ to obtain another
multi-graph $G'$.

\begin{definition}[multi-graph Restricted Swap Operation (mRSO)]\label{def:mrso}
  Let $G = (L \cup R, E)$ be a multi-graph, $a$ and $b$ be two
  non-necessarily distinct vertices in $L$, and $c$ and $d$ be two distinct
  vertices in $R$, such that there exist a port $x$ of $a$ and a port $y$ of $b$
  such that
  \[
    \{ (a,x,c), (b,y,d) \} \subseteq E \wedge  ( \degree{a} = \degree{b} \vee
    \degree{c} = \degree{d}) \enspace.
  \]
  The mRSO $(a,x,c), (b,y,d) \rightarrow (a,x,d), (b,y,c)$ is an operation that
  transforms $G$ into the multi-graph $G' = (L \cup R, E')$ such that $E' = (E
  \setminus \{ (a,x,c), (b,y,d) \}) \cup \{ (a,x,d), (b,y,c) \}$.
\end{definition}

It is easy to see that the multi-graph $G'$ obtained by applying an mRSO to $G$
is such that $\bjdm{G'} = \bjdm{G}$. There are zero or one mRSO between any two
multi-graphs in $\graphs$.
As an example, the mRSO $(\alpha,1,C), (\beta,4,E) \rightarrow (\alpha,1,E), (\beta,4,C)$ transforms the graph in \cref{fig:mrso} (left) to the graph on the right of such figure.
Here, patterns denote the side of nodes on the graph, and colors denote different degrees.
Dotted edges are the ones involved in the mRSO.%

\begin{figure}[htb]
  \centering
  \includegraphics[width=\textwidth]{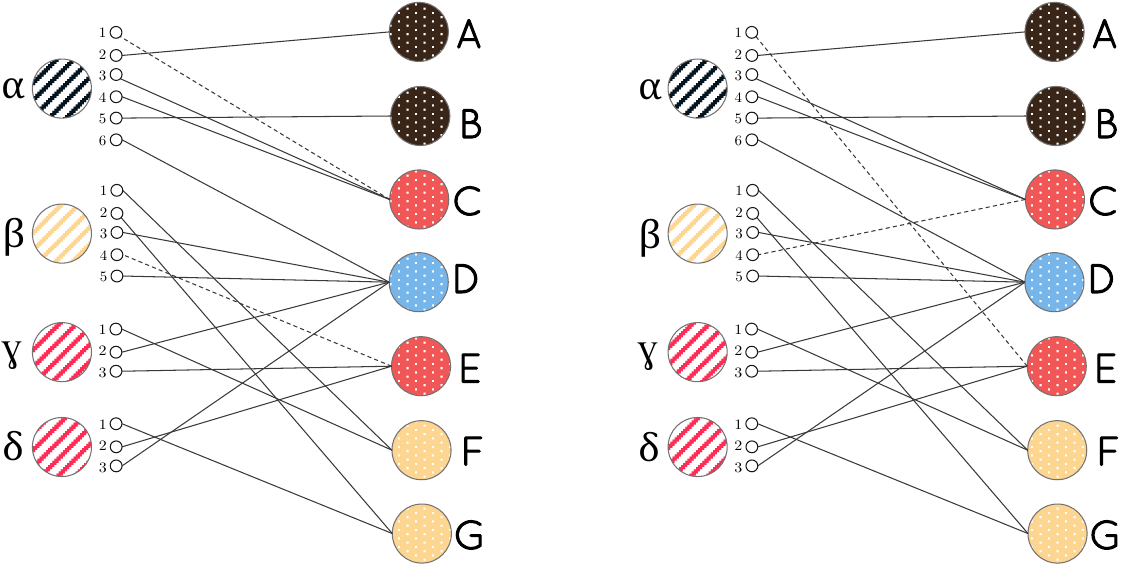}
  \caption{Example of an mRSO\@. Dotted edges are edges involved in the mRSO, different patterns denote nodes on different sides of the graph, and different colors denote different degrees.}\label{fig:mrso}
\end{figure}

The neighborhood structure of the state space $\graphs$ is such that there is an
edge from a multi-graph $G$ to a multi-graph $G'$ iff there is an mRSO
transforming $G$ into $G'$. In addition to these edges, there is a self-loop
from each state to itself. This structure results in a strongly connected space,
as can be seen by straightforwardly adapting \citep[Thm.\
8]{czabarka2015realizations} in a way similar to what was done also for the
bipartite simple graph case discussed in \cref{sec:sampling:swaps}.

We now move to defining the neighbor sampling distribution $\neighdistr{G}$ that
is used to propose the next state $G' \in \neigh{G}$ when the chain is at state
$G \in \graphs$. As in \cref{sec:sampling:swaps}, we first describe how to
sample a neighbor of $G$, and then analyze the resulting distribution.

For any $1 \le m \le \card{R}$ (resp.~$1\le n\le \card{\odataset}$), let $A_m$
(resp.~$B_n)$ be the subset of $L$ (resp.~of $R$) containing all and only the
vertices with degree $m$ in $G$ (but really, in any $G' \in \graphs$). The first
operation to sample a neighbor of $G$, is flipping a fair coin. If the outcome
is \emph{heads}, then we sample a degree $m$ proportional to the number of pairs
of \emph{not-necessarily-distinct} vertices in $L$ with degree $m$, i.e., we
draw  $1 \le m \le \card{R}$ with probability
\[
  \beta(m) = \slfrac{\binom{\card{A_m} + 1}{2}}{\sum_{j=1}^{\card{R}}
  \binom{\card{A_j} + 1}{2}}
\]
and then we draw two vertices $a$ and $b$ by sampling uniformly at random, with
replacement, from $A_m$. By sampling with replacement, we ensure that $a$ and
$b$ may be the same vertex. Consider now the set
\[
  H_{a,b} \doteq \{((a,x,f), (b,y,g)) \suchthat (a,x,f) \in E \wedge (b,y,g) \in
  E \wedge f \neq g\}
\]
of pairs of edges one incident to $a$ and one incident to $b$ and with different
endpoints in $R$, and sample a pair $((a,x,c), (b,y,d))$ uniformly at random
from this set.

If the outcome of the fair coin flip is \emph{tails}, we first sample a degree
$1 \le n \le \card{L}$ proportional to the number of pairs of \emph{distinct}
vertices in $R$ with degree $n$, i.e., we draw $1 \le n \le \card{L}$ with
probability
\[
  \gamma(n) = \slfrac{\binom{\card{B_n}}{2}}{\sum_{j=1}^{\card{\odataset}}
  \binom{\card{B_j}}{2}},
\]
and then we sample two \emph{distinct} vertices $c$ and $d$ from $B_n$ uniformly
at random without replacement. Let now $(a,x,c)$ (resp.~$(b,y,d)$) be an edge
sampled uniformly at random from those incident to $c$ (resp.~to $d$).

The mRSO $(a,x,c), (b,y,d) \rightarrow (a,x,d), (b,y,c)$, when applied to $G$,
gives the neighbor $G'$ which is the proposed next state for the Markov chain.

We now analyze the distribution $\neighdistr{G}$ over $\neigh{G}$ induced by
this procedure. Let $(a,x,c), (b,y,d) \rightarrow (a,x,d),
(b,y,c)$ be the sampled mRSO, and let $G' \in \neigh{G}$ be the multi-graph
obtained by applying this mRSO to $G$. It must be $G' \neq G$. Recall that this
mRSO is the only one leading from $G$ to $G'$. Consider the following events:
\begin{align*}
  E_\mathrm{\ell} &\doteq \text{``} \degree{a} = \degree{b} = m \text{''};\\
  E_\mathrm{r} &\doteq \text{``} \degree{c} = \degree{d} = n \text{''}\enspace.
\end{align*}
There are three possible cases for $\neighdistr{G}(G')$:
\begin{itemize}
  \item If only $E_\mathrm{\ell}$ holds, then
    \begin{equation}\label{eq:neighdistrleft}
      \neighdistr{G}(G') = \frac{1}{2} \frac{1}{\sum_{j=1}^{\card{R}} \binom{\card{A_j} + 1}{2}}
      \frac{1}{H_{a,b}}\enspace.
    \end{equation}
  \item If only $E_\mathrm{r}$ holds, then
    \begin{equation}\label{eq:neighdistrright}
      \neighdistr{G}(G') = \frac{1}{2} \frac{1}{\sum_{j=1}^{\card{\odataset}} \binom{\card{B_j}}{2}}
      \frac{1}{n^2}\enspace.
    \end{equation}
\item If both $E_\mathrm{\ell}$ and $E_\mathrm{r}$ hold, then
  $\neighdistr{G}(G')$ is the sum of the r.h.s.'s of
  \cref{eq:neighdistrleft,eq:neighdistrright}.
\end{itemize}

It is easy to see that $\neighdistr{G}(G') = \neighdistr{G'}(G)$, which, like
for \algorso, greatly simplifies the use of MH\@. As in that case, we define
$\statdistr$ over $\graphs$ as
\[
  \statdistr(G) \doteq \frac{\nullprob(\mattodat{G})}{\copiesnum{\mattodat{G}}},
\]
where $\copiesnum{\mattodat{G}}$ is still as in \cref{eq:copiesnum} because the
same result also holds for sequence datasets under the null model we are
considering~\citep[Lemma 4]{AbuissaLR23}. We can then conclude on the
correctness as follows, with the proof that is the same as that of
\cref{lem:swapcorrectness}.

\begin{lemma}\label{lem:seqcorrectness}
  Let $\dataset \in \nullset$. \algoseq outputs $\dataset$ with probability
  $\nullprob(\dataset)$.
\end{lemma}

\rev{\Cref{alg:alices} reports the operations performed by \algoseq to sample a sequence dataset in $\nullset$.
The algorithm receives in input the bipartite multi-graph $G \in \graphs$ corresponding to the observed sequence dataset $\odataset$ and a number of swaps $s$ sufficiently large for convergence.}
\begin{algorithm}[thb]
    \caption{\rev{\algoseq}}\label{alg:alices}
    \begin{algorithmic}[1]
    \Require Multi-Graph $G \in \graphs$, Number of Swaps $s$
    \Ensure Sequence Dataset $\dataset$ sampled from $\nullset$ with probability $\nullprob(\dataset)$ 
    \State $\copiesnum{\mattodat{G}} \gets $\Cref{eq:copiesnum}
    \State $i \gets 0$
    \While{$i < s$}
      \State $i \gets i + 1$
      \State $\mathsf{out} \gets$ flip a fair coin
      \If{$\mathsf{out}$ is \emph{heads}}
        \State $a$, $b$ $\gets$ vertices in $L$ drawn u.a.r. such that $\degree{a} = \degree{b}$ 
        \State $(a,x,c)$, $(b,y,d)$ $\gets$ pair drawn u.a.r. from $H_{ab}$
      \Else
        \State $c$, $d$ $\gets$ different vertices in $R$ drawn u.a.r. such that $\degree{c} = \degree{d}$
        \State $(a,x,c)$, $(b,y,d)$ $\gets$ edges drawn u.a.r. from those incident to $c$,$d$
      \EndIf
      \State $G' \gets $ perform $(a, x, c), (b, y, d) \rightarrow (a, x, d), (b, y, c)$ on $G$
      \State $\copiesnum{\mattodat{G'}} \gets $\Cref{eq:copiesnum}
      \State $p \gets $ random real number in $[0,1]$
      \State $a \gets \min\left(1, \copiesnum{\dataset}/\copiesnum{\dataset'}\right)$
      \If{$p \leq a$}
        $G \gets G'$
      \EndIf
    \EndWhile
    \State \Return $\mattodat{G}$
    \end{algorithmic}
\end{algorithm}

We leave for future work the development of a Curveball-like approach for
sampling sequence datasets from the null model. \Citet{JenkinsWGR22} propose
other two null models for sequence datasets. Extending these null models to
also preserve the BJDM is an interesting direction for future work.

% !TEX root = ../main-kais.tex
\section{Experimental Evaluation}\label{sec:exper}

We now report on the results of our experimental evaluation of \rev{\algorso, \algocurve, and \algoseq}.
% , as the main takeaways of
% the evaluation of this algorithm would be similar to those for \algorso, given
% the many similarities between the two algorithms, as discussed in
% \cref{sec:seq:algo}.
Our evaluation pursues three goals: empirically study the mixing time of
the sampling algorithms, evaluate their scalability as the number of 
transactions increases, and show that the null model we introduce differs from
that which only preserves the two fundamental properties, by showing that it
leads to marking different hypotheses as significant.

\begin{table}[htb]
\centering
\scriptsize
\caption{Datasets statistics: num.\ of transactions, num.\ of items, sum of
transaction lengths, avg.\ transaction length, density, and number of caterpillars.}%
\label{tbl:datasets}
	\begin{tabular}{lrrrrrr}
	\toprule
	Dataset &  Trans. &  Item &  Sum Trans. &  AVG Trans. &  Density & \rev{Num.}\\
	        & Num & Num & Lengths & Length & & \rev{Cater.} \\
	\midrule
	iewiki & 137 & 558 & 651 & 4.752 &   0.0085 & \rev{10K}\\
	kosarak & 3000 & 5767 & 23664 & 7.888 &   0.0014 & \rev{88M}\\
	chess & 3196 & 75 & 118252 & 37.000 &   0.4933 & \rev{9.93B}\\
	foodmart & 4141 & 1559 & 18319 & 4.424 &   0.0028 & \rev{954K}\\
	db-occ & 10000 & 8984 & 19729 & 1.973 &   0.0002 & \rev{7.5M}\\
	BMS1 & 59602 & 497 & 149639 & 2.511 &   0.0051 & \rev{1.13B}\\
	BMS2 & 77512 & 3340 & 358278 & 4.622 &   0.0014 & \rev{1.96B}\\
	retail & 88162 & 16470 & 908576 & 10.306 &   0.0006 & \rev{60B}\\
	\midrule
	\rev{SIGN} & \rev{730} & \rev{269} & \rev{76646} & \rev{104.994} & \rev{0.3903} & \rev{696M} \\
	\rev{LEVIATHAN} & \rev{5834} & \rev{9027} & \rev{400336} & \rev{68.621} & \rev{0.0076} & \rev{22B} \\
	\rev{FIFA} & \rev{20450} & \rev{2992} & \rev{1502634} & \rev{73.478} & \rev{0.0246} & \rev{159B} \\
 	\rev{BIKE} & \rev{21078} & \rev{69} & \rev{327844} & \rev{15.554} &  \rev{0.2254} & \rev{5.88B}\\
 	\rev{BIBLE} & \rev{36369} & \rev{13907} & \rev{1610501} & \rev{44.282} & \rev{0.0032} & \rev{259B} \\
  \rev{BMS1} & \rev{59601} & \rev{499} & \rev{358877} & \rev{6.021} & \rev{0.0121} & \rev{1.13B}\\
	\bottomrule
	\end{tabular}
\end{table}

\spara{Datasets.}
We use eight real-world transactional datasets \rev{and six real-world sequence datasets},\footnote{From \url{www.philippe-fournier-viger.com/spmf/index.php} and \url{http://konect.cc/networks}.} listed in \cref{tbl:datasets}.
Density is the ratio between the average transaction length and the number of items.
\textbf{iewiki} is a user-edit dataset, where each transaction is a set of Wikibooks pages edited by the same user;
\textbf{kosarak}, \textbf{BMS1}, \textbf{BMS2}, \rev{and \textbf{FIFA}} are click-stream datasets;
\textbf{chess} is a board-description datasets adapted from the UCI Chess (King-Rook vs King-Pawn) dataset;
\textbf{foodmart} and \textbf{retail} are retail transaction datasets;
\textbf{db-occ} includes user occupations taken from dbpedia;
\rev{\textbf{SIGN} is a dataset of sign language utterance;}
\rev{\textbf{LEVIATHAN} and \textbf{BIBLE} are sentence datasets created from the novel Leviathan by Thomas Hobbes (1651) and the Bible, respectively; and}
\rev{in \textbf{BIKE} each sequence indicate the bike sharing stations where a bike was parked in Los Angeles over time.}

\spara{Experimental Environment.}
We run our experiments on a $40$-Core ($2.40$ GHz) Intel\textregistered\ 
Xeon\textregistered\ Silver 4210R machine, with 384GB of RAM, and running
FreeBSD 14.0. Results are compared against GMMT~\citep{GionisMMT07}, which is a
swap randomization algorithm that samples from the null model that only
maintains the two fundamental properties. 
\rev{The sampler GMMT-S is a variant of the SelfLoop version of GMMT that preserves the left and right degree sequences of the bipartite multi-graph representation of the observed sequence dataset.}
All the samplers are implemented in Java 1.8, and the code is available at
\url{https://github.com/acdmammoths/alice}.

\begin{figure}[thb]
	\centering
	\includegraphics[width=.9\textwidth]{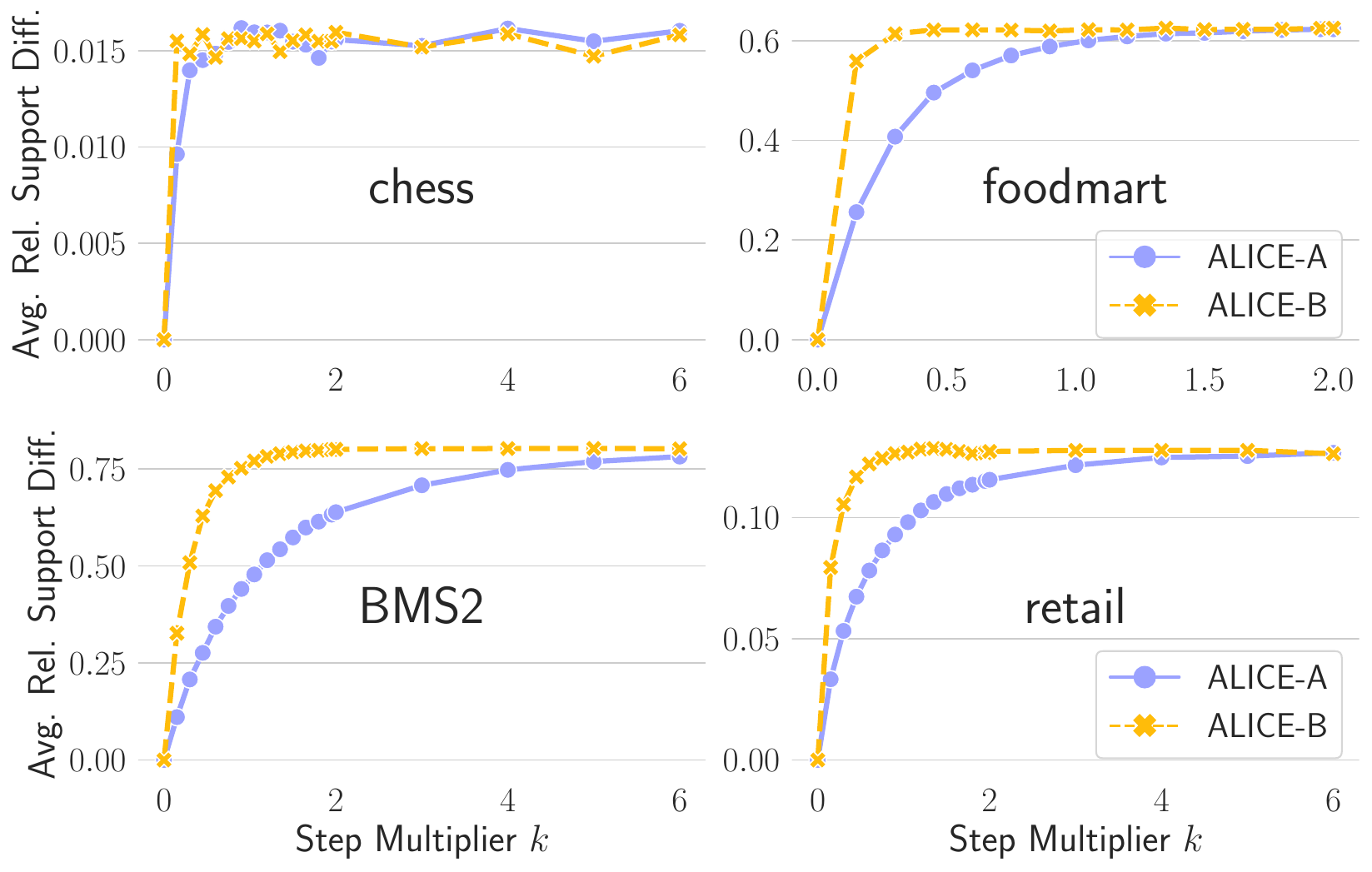}
	\caption{\rev{Convergence of the samplers increasing the step number multiplier $k$,
	for chess (upper left), foodmart (upper right), BMS2 (lower left), and retail (lower right).}}%
	\label{fig:convergence}
\end{figure}

\begin{figure}[thb]
	\centering
	\includegraphics[width=\textwidth]{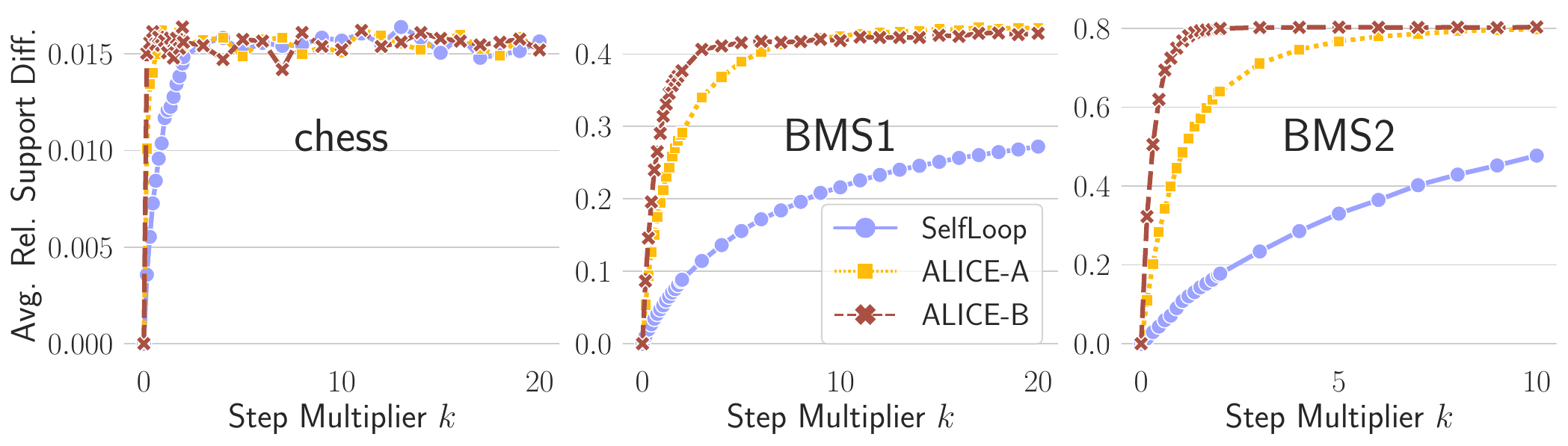}
	\caption{\rev{Convergence of \algorso and \algocurve vs SelfLoop increasing the step number multiplier $k$,
	for chess (left), BMS1 (middle), and BMS2 (right).}}%
	\label{fig:convergence_self_loop}
\end{figure}

\spara{Convergence.}
To study the convergence of our samplers, we follow a procedure similar to the 
one proposed by \citet{GionisMMT07}. The mixing time, i.e., the number of steps
needed for the state of the chain to be distributed according to $\nullprob$, is
estimated by looking at the convergence of the \emph{Average Relative Support
Difference (ARSD)}, defined as 
\[
ARSD(\dataset^s) = \frac{1}{\card{\fis{\odataset}{\thresh}}}\sum_{A \in
\fis{\odataset}{\thresh}}\frac{\abs{\supp{\odataset}{A} - \supp{\dataset^s}{A}}}{\card{\supp{\odataset}{A}}},
\]
where $\dataset^s$ is the dataset obtained by the sampler after $s$ steps.
\Cref{fig:convergence} reports this quantity for \rev{chess (upper left)}, foodmart (upper right), BMS2 (lower left), and retail (lower right), for $s = \lfloor k \cdot w \rfloor$ with $k \in \{0,
0.15, 0.3, \dots, 2, 3, \dots, 6\}$ and $w = \sum_{t \in \odataset}\card{t}$.
Results for other datasets were qualitatively similar. %
%The parameter $k$ is a step multiplier factor, used to set the number of steps $s$ 
%proportional to the sum of transaction lengths of the dataset.
\algocurve needs $\sfrac{1}{3}$ or even fewer steps than \algorso, thanks
to to the fact that it essentially performs multiple RSOs at each step (as each
RBSO corresponds to one or more RSOs).
%In particular, ARSD stabilizes roughly at 0.45, 1.5, and 1.05, for \algocurve in foodmart, 
%BMS2, and retail; whereas for \algorso the values are 1.5, 6, and 5, respectively.

Despite the fewer number of \emph{steps} needed, the \emph{(wall clock) time} to
convergence of \algocurve (not reported in figures), however, is higher than
that of \algorso. This difference is due to the fact that performing an RBSO, which is a
more complex operation than an RSO, requires additional bookkeeping for each
element in the set $U$ (see \cref{def:rbso}). In the worst cases (BMS1, and
chess), \algocurve takes almost 10x the time of \algorso to reach
convergence. An interesting direction for future work is to study how to avoid
this additional bookkeeping in \algocurve to obtain the same advantage over
\algorso observed for the number of steps to convergence also for the wall clock time.

\rev{\Cref{fig:convergence_self_loop} compares the ARSD values obtained by
\algo with those measured in the states of the chain traversed by the na\"{\i}ve approach introduced in \Cref{sec:sampling:swaps} (called SelfLoop in the figure). Recall that, at each step, this approach draws two pairs $(a,c)$ and $(b,d)$ of row-column indices uniformly at random, and moves to the next state if $(a,c), (b,d) \rightarrow (a,d), (b,c)$ is a RSO.}
\rev{Especially for larger datasets, we observe that SelfLoop moves slowly in the state space, which prevents the ARSD from stabilizing even after $10w$ steps.}
\rev{As a result, a large number of steps is required to increase the likelihood of convergence, thus rendering SelfLoop impractical for use.}
\rev{In fact, the running time increases with the number of steps. In BMS1, for example, the ARSD for \algocurve stabilizes around $k=4$, with \algocurve taking roughly 17s to perform the $4w$ steps. In contrast, SelfLoop takes 397s to perform the $10w$ steps.}

\begin{figure}[thb]
	\centering
	\includegraphics[width=.9\textwidth]{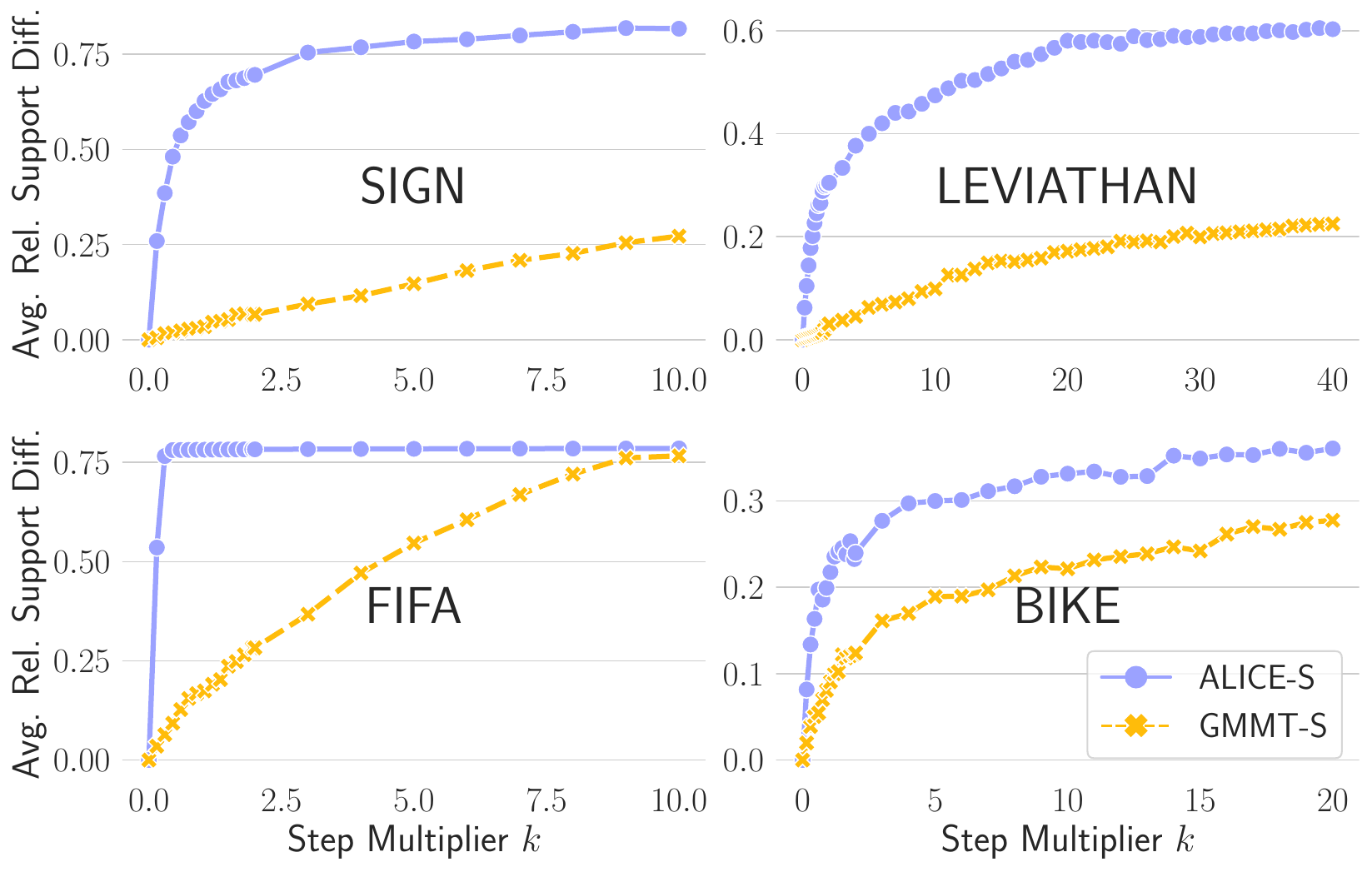}
	\caption{\rev{Convergence of \algoseq and GMMT-S increasing the step number multiplier $k$,	for SIGN (upper left), LEVIATHAN (upper right), FIFA (lower left), and BIKE (lower right).}}%
	\label{fig:convergence_seq}
\end{figure}  

\rev{We notice a similar behavior in \Cref{fig:convergence_seq}, which illustrates the convergence of \algoseq and GMMT-S for the sequence datasets SIGN (upper left), LEVIATHAN (upper right), FIFA (lower left), and BIKE (lower right). In this case, $w=|E|$, i.e. the number of edges in the multi-graph corresponding to the dataset. In SIGN and FIFA the ARSD stabilizes before $k=3$ for \algoseq, whereas for GMMT-S it stabilizes only in FIFA. In BIKE and LEVIATHAN both samplers move slowly, and thus convergence is reached after almost $20w$ and $30w$ steps, respectively.}

\begin{figure}[htb]
	\centering
	\includegraphics[width=.8\textwidth]{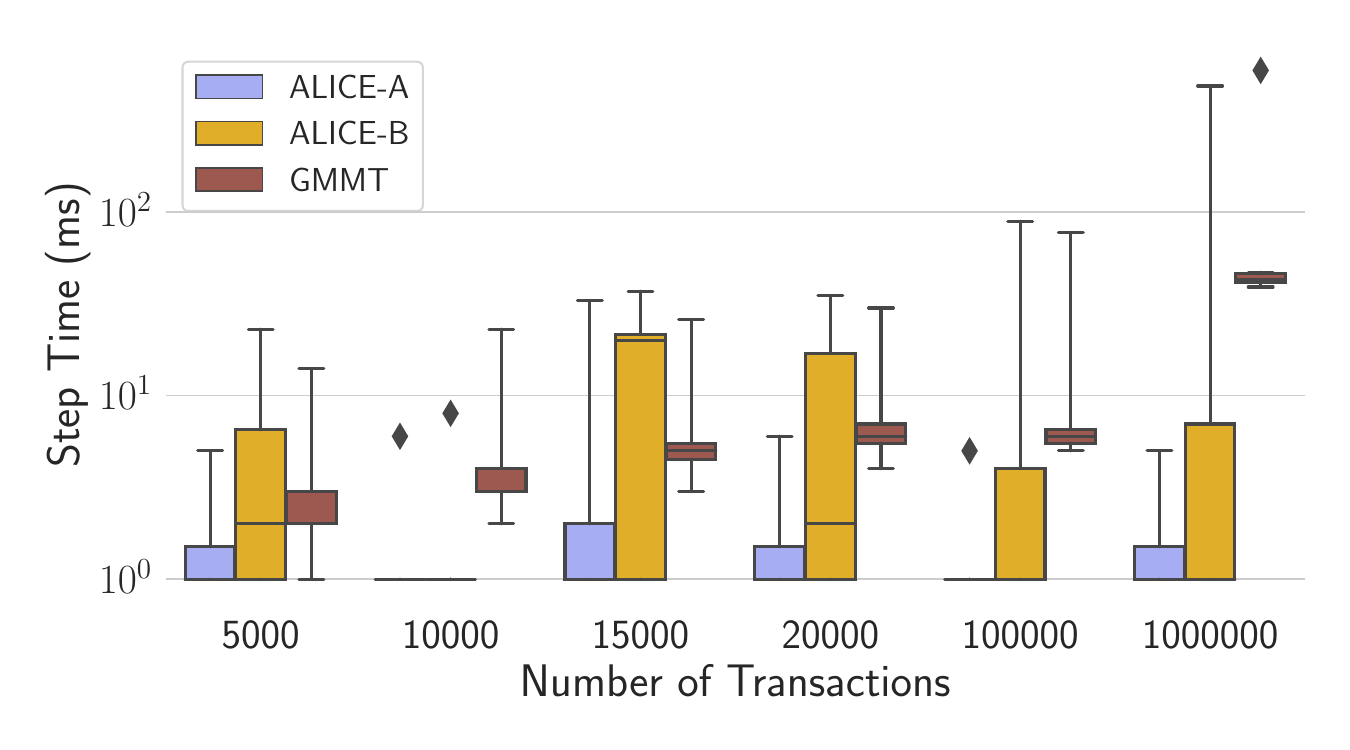}
	\caption{\rev{Step times of the samplers in the synthetic datasets.}}%
	\label{fig:scalability}
\end{figure}

\spara{Scalability.}
\rev{To study the scalability of \algo, we create synthetic datasets with
increasing number of transactions and
average transaction length $25$, by using the IBM Quest
generator~\citep{AgrawalS94}: five datasets with $100$ items and $5$k, $10$k, $15$k, $20$k, and $100$k transactions, and one dataset with $10$k items and $1$M transactions.} 
For each sampler \rev{for transactional datasets}, we perform 10k steps and
compute the distribution of step times, reported in \cref{fig:scalability} (log values). 
For
completeness, we include the step times of GMMT, although they are not really comparable
to those of our algorithms, because GMMT samples from a different null
set $\nullset$ which includes datasets with different BJDMs. The median step
time scales linearly with the size of the dataset. \algorso is the fastest
sampler, requiring less than 8ms to perform a step in the largest dataset, and
less than 1ms in most of the cases. In contrast, the step times of \algocurve
are characterized by more variability, as they depend on \emph{(i)} whether
the performed RBSO is an rRBSO or a cRBSO, and \emph{(ii)} the size of the set
$U$: the time required to compute $\copiesnum{\dataset}$ is larger for cRBSO,
and it grows with the size of $U$. %, as discussed above about the wall clock step time for \algocurve.

\begin{figure}[htb]
	\centering
	\includegraphics[width=.8\textwidth]{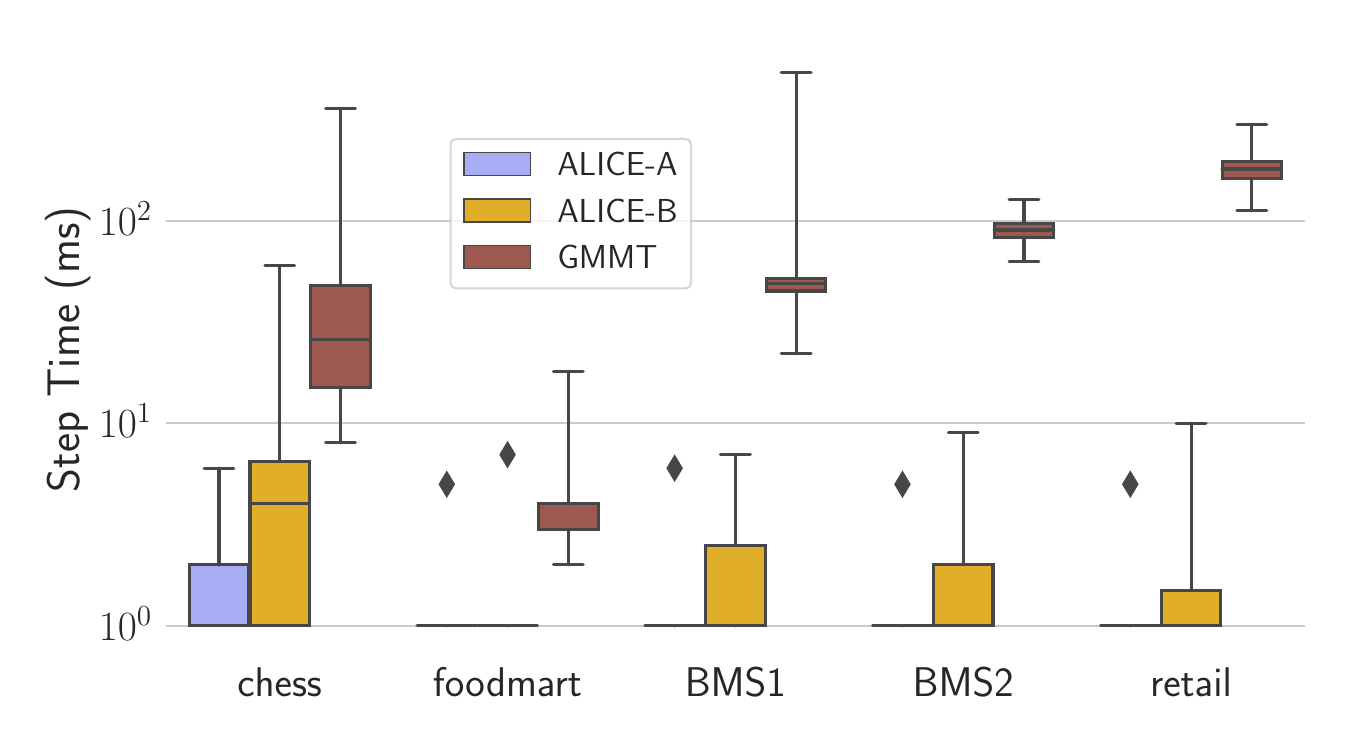}
	\caption{\rev{Step times of the samplers in the real datasets (log times).}}
	\label{fig:scalability_real}
\end{figure}

\begin{figure}[htb]
	\centering
	\includegraphics[width=.8\textwidth]{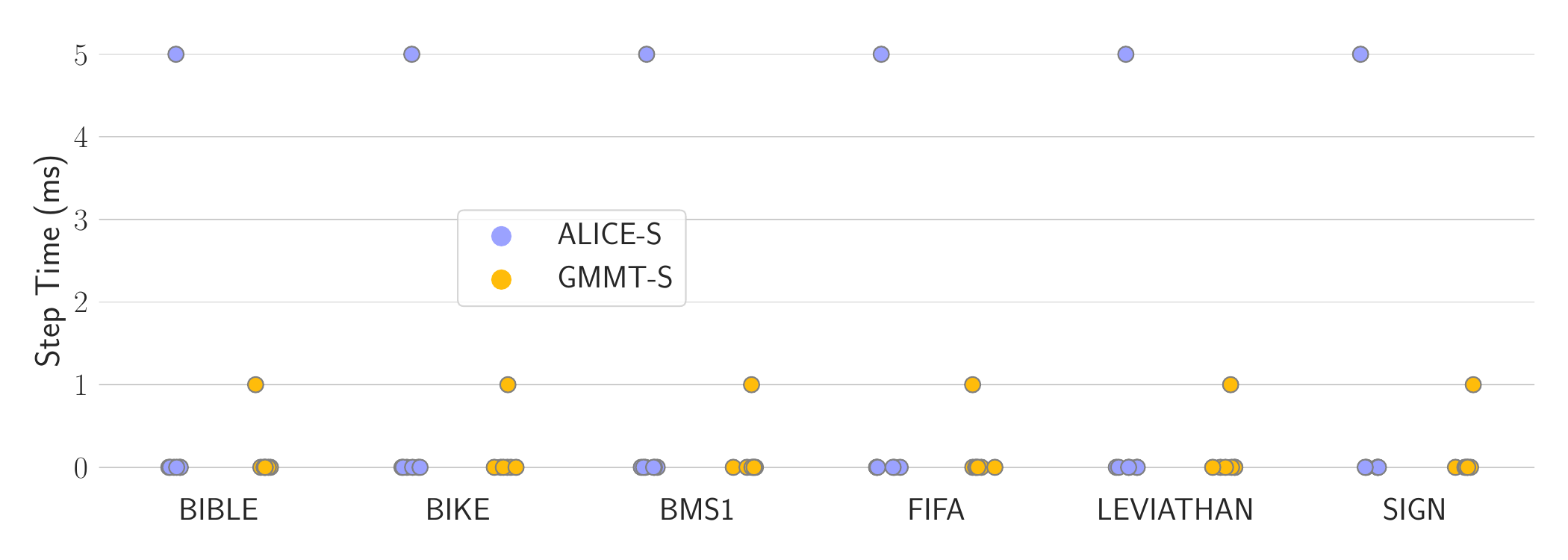}
	\caption{\rev{Step times of the samplers in the real sequence datasets.}}%
	\label{fig:scalability_seq}
\end{figure}

\rev{\Cref{fig:scalability_real} reports the distribution of the time required to perform a step}
% min, Q1, median, Q3, and max time required to perform a step, 
\rev{for each sampler in each transactional dataset}. The step time of \algocurve tends to be larger 
in chess, despite
it not being the largest dataset. This fact is due to the high density of this dataset,
and its large transaction length (37). Hence, the size of $U$ is usually high.
In foodmart, on the other hand, the average
transaction length is $4.42$ and the average item support is $5.6$, so the size
of $U$ is often $1$.
An algorithmic improvement in the bookkeeping due to the size of $U$
would results in better performance of \algocurve, as mentioned above.

\rev{\Cref{fig:scalability_seq} shows the distribution of step times for \algoseq and GMMT-S in the sequence datasets. The performance of \algoseq is comparable with that of \algorso, as they follow a similar approach to sample the swap operations to perform. The median step time is always $<1$, and the algorithm takes at most $5$ms to perform a step. The step times of GMMT-S are far lower than its counterpart for transactional datasets, because this algorithm does not require bookkeeping to compute the transition acceptance probability. we recall that also in this case the running time of GMMT-S is not really comparable with that of our algorithm because they sample from different null models.}  
% \begin{figure}[t!]
% 	\centering
% 	\begin{minipage}[b]{.49\linewidth}
% 		\centering
% 		\includegraphics[width=\textwidth]{figures/iewiki_numSigItemsets.pdf}
% 	\end{minipage}
% 	\begin{minipage}[b]{.49\linewidth}
% 		\centering
% 		\includegraphics[width=\textwidth]{figures/chess_numSigItemsets.pdf}
% 	\end{minipage}
% 	\caption{Number of significant itemsets per length for \algorso, \algocurve, and GMMT, 
% 	in iewiki (left) and chess (right).}
% 	\label{fig:signfi}
% \end{figure}

\spara{Significance of the Number of Frequent Itemsets.}
To show that the null model we introduce is different than the one that only
preserves the two fundamental properties, We test the null hypothesis $H_0$
from \cref{eq:nullhyp}, and estimate the $p$-value as in \cref{eq:epval} with
$T=4352$ samples from the null model, for each sampler.\footnote{The number of
steps is empirically fixed according to the results obtained in the convergence
experiment.} We remark that this kind of hypothesis is just a simple but clear
example of the tasks that can (and should) be formed to assess the statistical
validity of results obtained from transactional datasets. Other tasks include,
for example, mining the statistically-significant frequent itemsets. We limit
ourselves to this task because it is straightforward to present and it is
sufficient to show the significant (pun intended) difference between preserving
the BJDM, as our null model does, and not preserving it.

\begin{table}[htb]
	\centering
	\scriptsize
	\caption{No.\ of FIs in the original dataset $\odataset$, avg.\ no.\ of 
	FIs in the sample $\dataset_i$, estimated p-value $\epval{\odataset, H_0}$ for $H_0$ from \cref{eq:nullhyp}.}%
	\label{tab:freqsignit}
	\begin{tabular}{crlrr}
		\toprule
    Dataset & $\card{\fis{\odataset}{\thresh}}$ & Sampler &
    $\frac{\sum_1^T{\card{\fis{\dataset_i}{\thresh}}}}{T}$ & \epval{\odataset, H_0}\\
		\midrule
		\multirow{3}{*}{iewiki} & \multirow{3}{*}{65665}
		& \algorso & 173 & $2.3$E-$4$ \\
		\multirow{3}{*}{$\thresh=1.4$E-$2$} & & \algocurve & 171 & $2.3$E-$4$ \\
		& & GMMT & 2257 & $1.8$E-$2$ \\
		\midrule
		\multirow{3}{*}{kosarak} & \multirow{3}{*}{6277}
		& \algorso & 4865 & $2.3$E-$4$ \\
		\multirow{3}{*}{$\thresh=3.0$E-$3$} & & \algocurve & 4130 & $2.3$E-$4$ \\
		& & GMMT & 31774 & $1.0$E-$0$ \\
		\midrule
		\multirow{3}{*}{chess\footnote{For chess and BMS1, $T=2176$, due to the prohibitive running time of GMMT.}} & \multirow{3}{*}{8227}
		& \algorso & 6183 & $4.6$E-$4$ \\
		\multirow{3}{*}{$\thresh=0.8$} & & \algocurve & 6182 & $4.6$E-$4$ \\
		& & GMMT & 6179 & $4.6$E-$4$ \\
		\midrule
		\multirow{3}{*}{foodmart} & \multirow{3}{*}{4247}
		 & \algorso & 2229 & $2.3$E-$4$ \\
		\multirow{3}{*}{$\thresh=3.0$E-$4$} & & \algocurve & 2228 & $2.3$E-$4$ \\
		& & GMMT & 2226 & $2.3$E-$4$ \\
		\midrule
		\multirow{3}{*}{db-occ} & \multirow{3}{*}{834} 
		& \algorso & 702 & $2.3$E-$4$ \\
		\multirow{3}{*}{$\thresh=5.0$E-$4$} & & \algocurve & 703 & $2.3$E-$4$ \\
		& &    GMMT & 598 & $2.3$E-$4$ \\
		\midrule
		\multirow{3}{*}{\rev{BMS1}} & \multirow{3}{*}{\rev{3991}} 
		& \rev{\algorso} & \rev{1998} & \rev{$4.6$E-$4$}\\
		\multirow{3}{*}{\rev{$\thresh=0.001$}} & & \rev{\algocurve} & \rev{1609} & \rev{$4.6$E-$4$}\\
		& &    \rev{GMMT} & \rev{1800} & \rev{$4.6$E-$4$}\\
		\bottomrule
	\end{tabular}
\end{table}

\Cref{tab:freqsignit} reports the
number of FIs in the observed dataset, the average number of FIs in the sampled
datasets, and the empirical $p$\nobreakdash-value, for datasets where GMMT
terminated within two days. The fact that (very) different $p$-values can be
obtained with \algo and with GMMT, which sample from a different null model,
highlights the striking impact of preserving the BJDM. As an example, for any
critical value in $(0.00023,0.01815)$, in iewiki $H_0$ would be rejected under the null model we introduce, but not under the null model that
only preserves the two fundamental properties. 

\begin{figure}[ht]
	\centering
		\includegraphics[width=.9\textwidth]{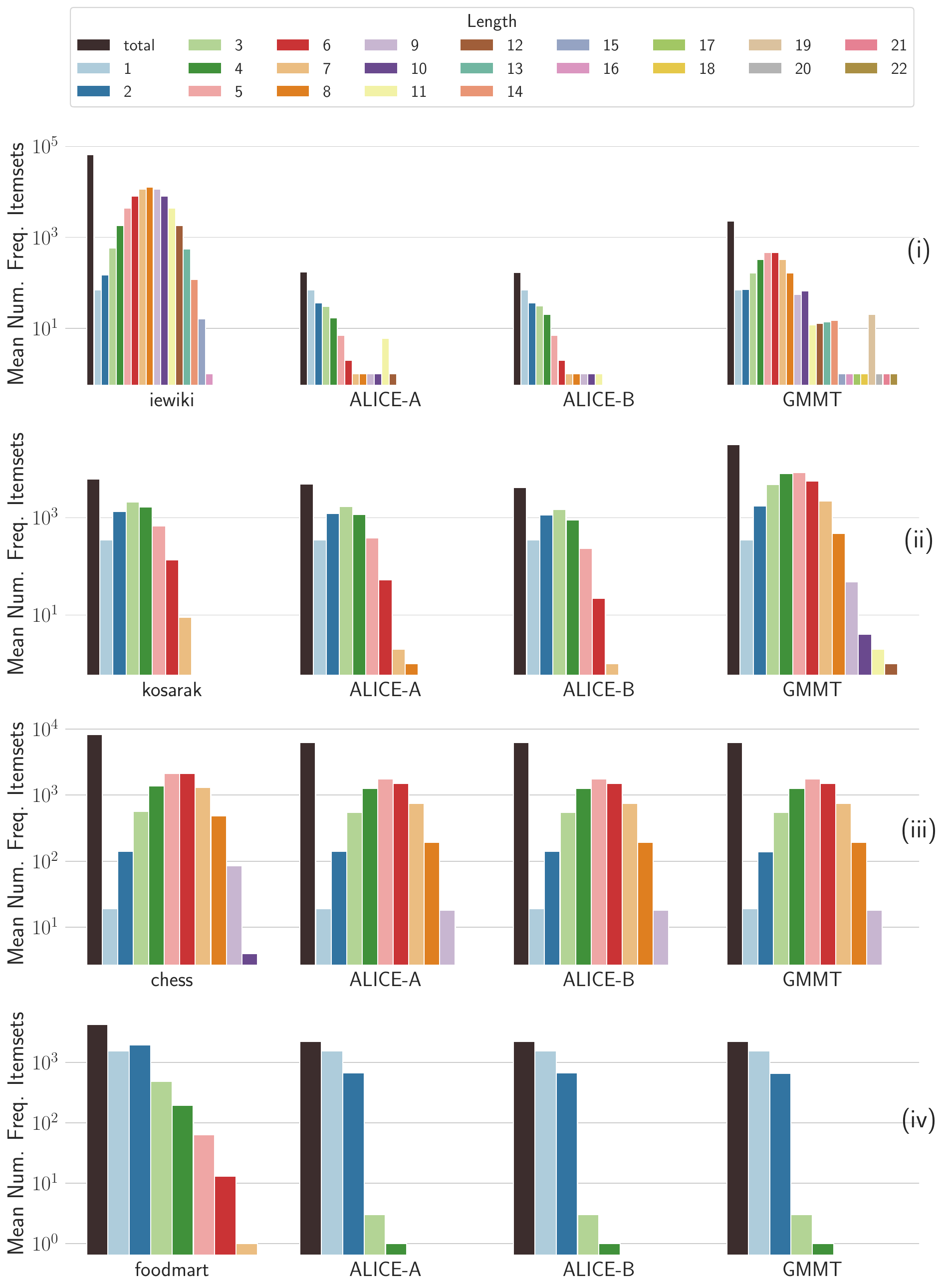}
	\caption{Mean number of frequent itemsets per length for \algorso, \algocurve, and GMMT, in iewiki (i), kosarak (ii), chess (iii), and foodmart (iv).}%
	\label{fig:freqit_S}
\end{figure}

\begin{figure}[ht]
	\centering
		\includegraphics[width=.9\textwidth]{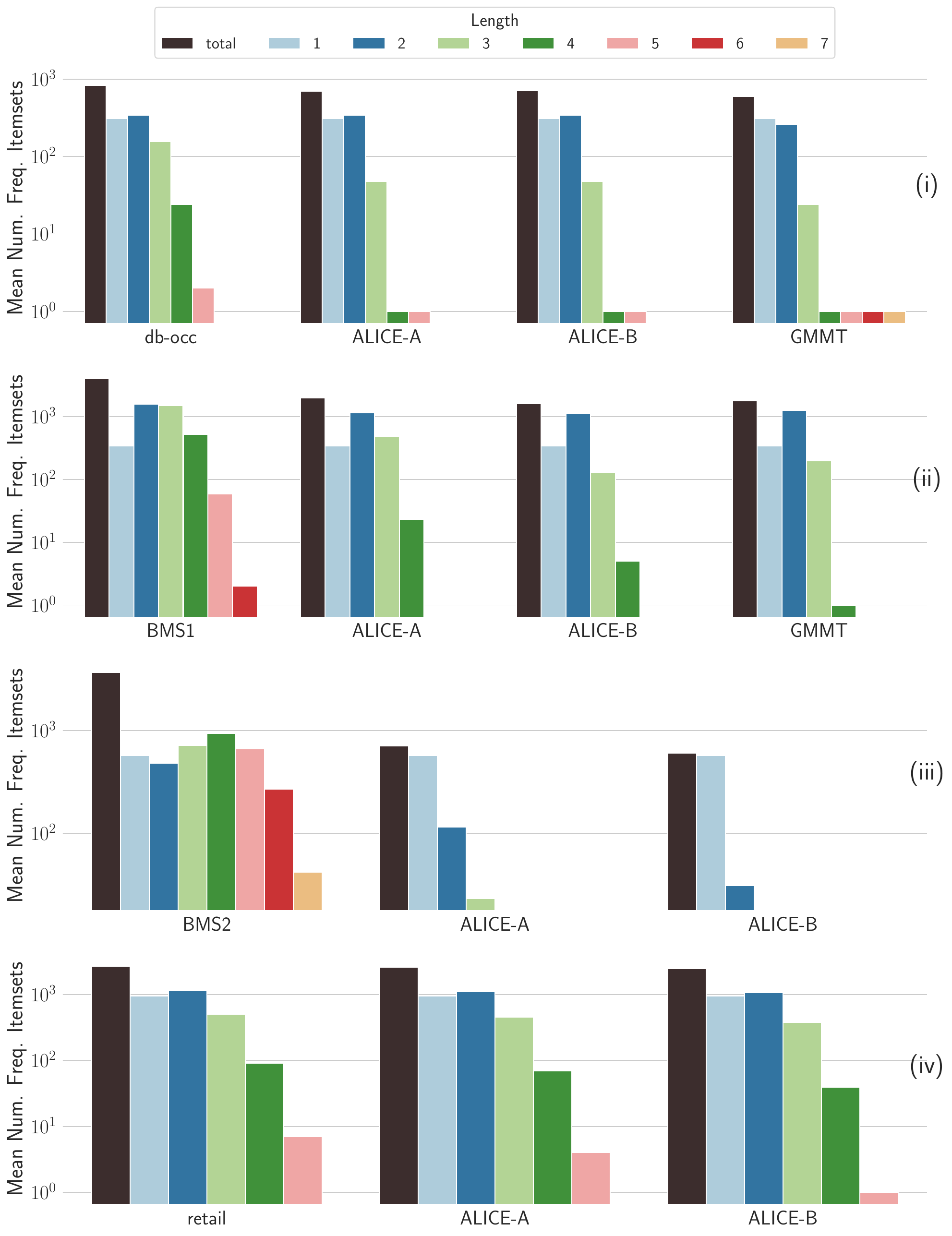}
	\caption{\rev{Mean number of frequent itemsets per length for \algorso, \algocurve, and GMMT (when available), in db-occ (i), BMS1 (ii), BMS2 (iii), and retail (iv).}}%
	\label{fig:freqit_L}
\end{figure}

\begin{figure}[thb]
	\centering
		\includegraphics[width=.9\textwidth]{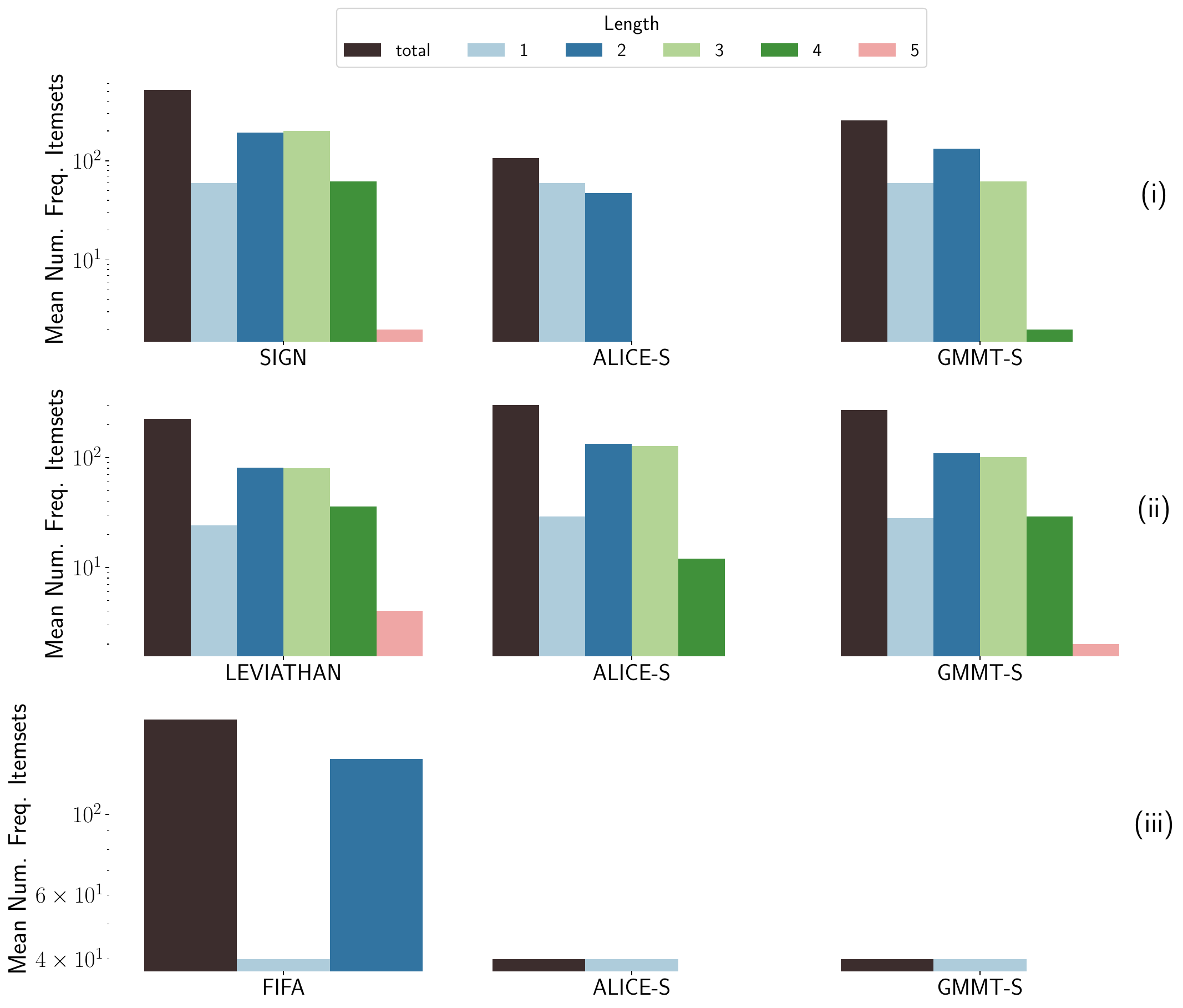}
	\caption{\rev{Mean number of frequent itemsets per length for \algoseq and GMMT-S, in SIGN (i), LEVIATHAN (ii), and FIFA (iii).}}%
	\label{fig:freqit_seq_S}
\end{figure}

\Cref{fig:freqit_S} and \Cref{fig:freqit_L} show the
distribution of the number of FIs of different lengths in the original datasets,
and the average of the same quantity over the datasets sampled by the different
samplers. 
For BMS2 and retail we do not report results for GMMT, due to its prohibitive running time.
Since they sample from the same null model, \algorso and \algocurve
obtain the same distribution (up to sampling noise), which is quite different
than the one obtained by GMMT. Note that whether the sampled
datasets have more or less FIs than the observed dataset depends both on the
null model and on the dataset.
For instance, in iewiki (\cref{fig:freqit_S}, i) datasets sampled from all null
models have fewer FIs than the observed one.
Conversely, in
kosarak (\cref{fig:freqit_S}, ii) the BJDM-preserving null model
produces samples with a similar number of FIs, while the datasets sampled from
the null model that preserves the  two fundamental properties have a larger number of FIs.
In addition, in iewiki, the samples from this latter model usually contain FIs
of length larger than any FIs in the observed dataset: the max length of a FI in iewiki is 16, whereas it grows to 22 in the datasets sampled by GMMT.
In kosarak, the datasets sampled by GMMT contain both a larger number of FIs per
length and FIs of larger length (12 vs. 7). The increase in the number of FIs of
length three, leads to a substantial difference in the number of FIs of length in
the range $[4,7]$: we observe up to 246x more FIs in the sampled datasets.
In contrast, since all the transactions in chess have the same length, we
observe (\cref{fig:freqit_S}, iii) similar average numbers of FIs across the
samplers. In this dataset, any swap operation performed by GMMT is actually a
RBSO, and hence also the datasets sampled by GMMT preserve the BJDM. Similarly,
the fact that the nodes in the graph representation of foodmart (\cref{fig:freqit_S}, iv) display high
assortativity indicates that most of the swap operations of GMMT are RBSO. 
In fact, when the product between the two marginals is close to the BJDM in terms of Frobenius norm, preserving the marginals \emph{almost} preserves the BJDM\@
As a consequence, also in this case, the
distribution of the numbers of FIs for GMMT is similar to that for \algo.

We can see that the distribution of the number of FIs in the observed dataset is always different from those obtained from the sampled datasets.
In particular, the longer itemsets are, in general, less frequent in the sampled datasets than in the original dataset.
As an example, BMS2 (\cref{fig:freqit_L}, iii) contains many FIs of length larger
than three (roughly 52\% of the FIs), while most of the FIs in the datasets
sampled by \algo have length one.

\rev{\Cref{fig:freqit_seq_S} and \Cref{fig:freqit_seq_L} present the distribution of frequent sequential itemsets of different lengths in the original sequence datasets, and the average of the same quantity over the datasets sampled by \algoseq and GMMT-S. The frequency thresholds used are taken from \cite{TononV19}: $0.4$ for SIGN, $0.15$ for LEVIATHAN, $0.275$ for FIFA, $0.025$ for BIKE, $0.1$ for BIBLE, and $0.002$ for BMS1. The number of samples extracted is always $4352$ and the number of steps performed by \algoseq is $10w$, while it is $50w$ for GMMT-S. Also in this case, $w$ is the number of edges in the multi-graph corresponding to the dataset. Similarly to the transactional dataset case, we tend to observe frequent itemsets of larger size in the datasets sampled by GMMT-S, except in the case of few frequent itemsets in the original dataset (e.g. FIFA and BIKE). In such cases, only trivial itemsets are frequent, and their frequencies tend to be preserved by preserving the two fundamental properties.}

\begin{figure}[thb]
	\centering
		\includegraphics[width=.9\textwidth]{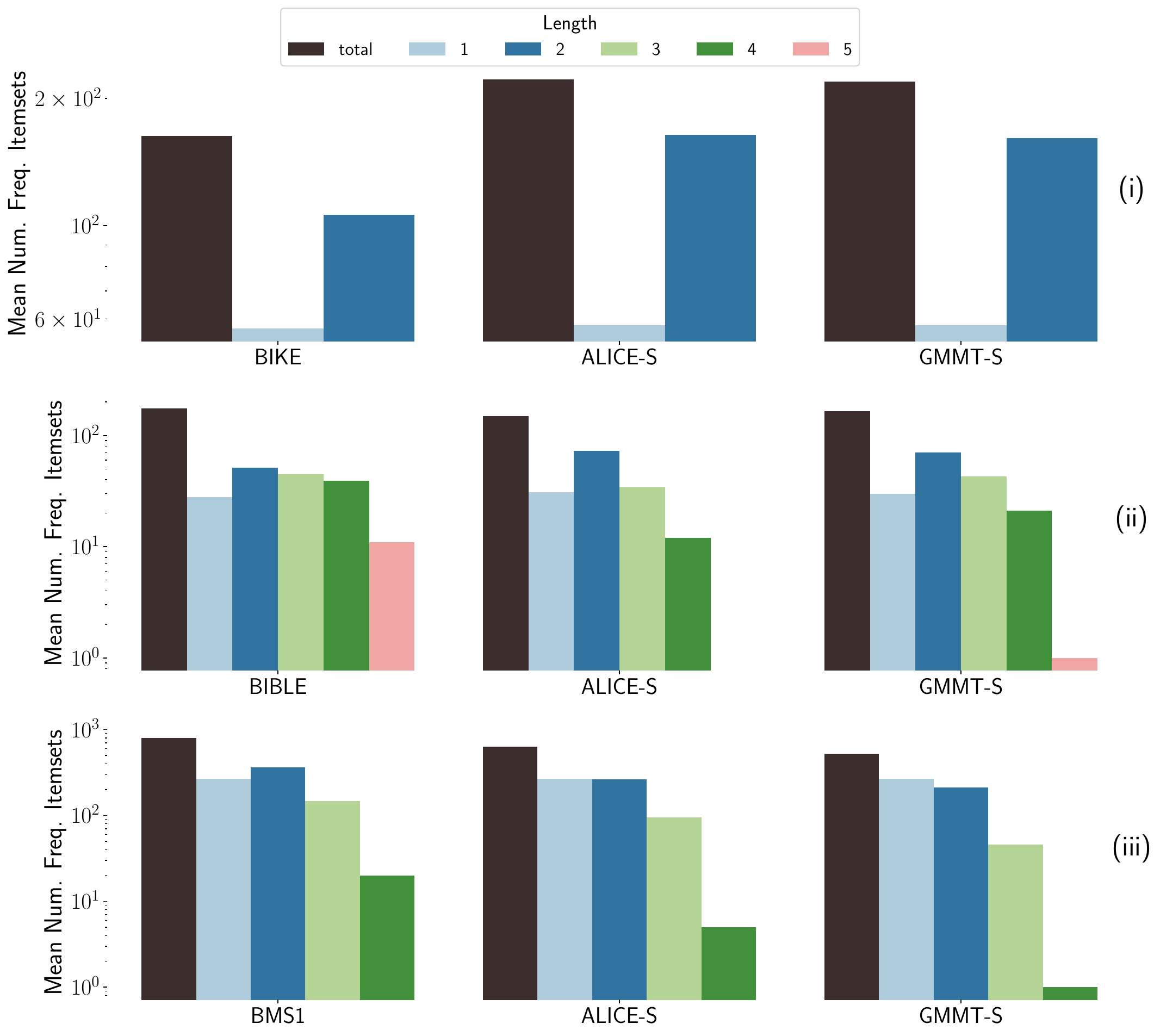}
	\caption{\rev{Mean number of frequent itemsets per length for \algoseq and GMMT-S, in BIKE (i), BIBLE (ii), and BMS1 (iii).}}%
	\label{fig:freqit_seq_L}
\end{figure}

Thanks to these results, we conclude that the BJDM captures important additional
information about the data generation process. Therefore, using a null model
that preserves it may lead to very different conclusions about the data
generation process compared to one that does not. These results highlight, once
more, how the choice of the null model by the user must be extremely deliberate.

\section{Conclusion}\label{sec:concl}

We introduced a novel null model for statistically assessing the results obtained
from an observed transactional or sequence dataset, preserving its Bipartite
Joint Degree Matrix (BJDM).  On transactional datasets, maintaining this
property enforces, in addition to the dataset size, transaction lengths, and
item supports, also the preservation of the number of \emph{caterpillars} of the
bipartite graph corresponding to the observed dataset, which is a natural and
important property that captures additional structure. We describe \algo, a
suite of Markov-Chain-Monte-Carlo algorithms for sampling datasets from the null
models. The results of our experimental evaluation show that \algo scales well
and that, when testing results w.r.t.\ our null models, different results are
marked as significant than when using existing null models.

A good direction for future work includes a rigorous theoretical analysis and/or
experimental evaluation of the trade-offs between the time taken to perform a single
step and the mixing time of the Markov chain when using different neighbor
sampling distribution. Towards making statistically-sound knowledge discovery a
reality, we also suggest the development of even more descriptive null models
(e.g., by preserving the number of
\emph{butterflies}~\citep{sanei2018butterfly}), and of efficient procedures to
sample from them, which is usually the challenging aspect.  Another interesting
direction is proposing null models for real-valued transactional datasets, such
as those used for high-utility itemsets mining.

\section*{Acknowledgments}
This work is sponsored in part by NSF award IIS-2006765.

\bibliography{refs}

\end{document}